\documentclass[a4paper,12pt,oneside]{article}
\usepackage{amsmath,amssymb,mathtools,bm}
\usepackage[utf8]{inputenc}
\usepackage{csquotes}
\usepackage[english]{babel}
\usepackage[T1]{fontenc}
\usepackage{amstext}
\usepackage{amsthm}
\usepackage{bbm}
\usepackage{enumerate}
\usepackage{hyperref}
\usepackage[nameinlink,capitalise]{cleveref}
\usepackage{braket}
\usepackage{dsfont}
\usepackage{color}
 \usepackage{tikz}
\usepackage{pgfplots}

\usepackage{authblk}

\usepackage[verbose=true,letterpaper]{geometry}
\AtBeginDocument{
	\newgeometry{
		textheight=9.5in,
		textwidth=7in,
		top=1in,
		headheight=14.5pt,
		headsep=10pt,
		footskip=25pt
	}
}

\usepackage{fancyhdr}
\fancyheadoffset{0pt}
\usepackage{titlesec}
\titleformat{\section}{\bfseries\scshape\Large}{\thesection}{1em}{}{}
\titleformat*{\subsection}{\scshape\bfseries\large}

\renewenvironment{thebibliography}[1]{
	\begin{oldthebibliography}{#1}
		\footnotesize
		\setlength{\itemsep}{0em}
		\setlength{\parskip}{0em}
	}
	{
	\end{oldthebibliography}
}

\numberwithin{equation}{section}
\newtheorem{thm}{Theorem}[section]
\newtheorem{lem}[thm]{Lemma}
\newtheorem{prop}[thm]{Proposition}
\newtheorem{cor}[thm]{Corollary}

\theoremstyle{definition}

\newtheorem{hyp}{Hypothesis}
\renewcommand*{\thehyp}{\Alph{hyp}}

\theoremstyle{remark}
\newtheorem{rem}[thm]{Remark}
\newtheorem{ex}[thm]{Example}

\crefname{hyp}{Hypothesis}{Hypotheses}
\Crefname{hyp}{Hypothesis}{Hypotheses}
\crefname{lem}{Lemma}{Lemmas}
\Crefname{lem}{Lemma}{Lemmas}
\crefname{thm}{Theorem}{Theorems}
\Crefname{thm}{Theorem}{Theorems}
\crefname{prop}{Proposition}{Propositions}
\Crefname{prop}{Proposition}{Propositions}
\crefname{enumi}{}{}
\Crefname{enumi}{}{}
\creflabelformat{enumi}{#2(#1)#3}
\crefname{equation}{}{}
\Crefname{equation}{}{}
\crefname{rem}{Remark}{Remarks}
\Crefname{rem}{Remark}{Remarks}

\makeatletter
\renewcommand{\@upn}{} % to use the same font for the number as for the head
\makeatother

\usepackage{etoolbox}
\makeatletter
\patchcmd{\endthm}{\@endpefalse}{}{}{}
\patchcmd{\endcor}{\@endpefalse}{}{}{}
\patchcmd{\endlem}{\@endpefalse}{}{}{}
\patchcmd{\endprop}{\@endpefalse}{}{}{}
\patchcmd{\endproof}{\@endpefalse}{}{}{}
\makeatother

\usepackage[inline]{enumitem}
\newlist{enumthm}{enumerate}{1} % set up a dedicated enumeration env.
\setlist[enumthm]{label=\upshape(\roman*),ref=\thethm~(\roman*)}  %labels are upshape, references as environment
\crefalias{enumthmi}{thm} % alias 'enumthmi' counter to 'thm'
\newlist{enumcor}{enumerate}{1}
\setlist[enumcor]{label=\upshape(\roman*),ref=\thecor~(\roman*)}
\crefalias{enumcori}{cor}
\newlist{enumlem}{enumerate}{1}
\setlist[enumlem]{label=\upshape(\roman*),ref=\thelem~(\roman*)}
\crefalias{enumlemi}{lem}
\newlist{enumprop}{enumerate}{1}
\setlist[enumprop]{label=\upshape(\roman*),ref=\theprop~(\roman*)}
\crefalias{enumpropi}{prop}
\newlist{enumhyp}{enumerate}{1}
\setlist[enumhyp]{label=\upshape(\roman*),ref=\thehyp~(\roman*)}
\crefalias{enumhypi}{hyp}
\newlist{enumproof}{enumerate*}{1}
\setlist[enumproof]{label=\upshape(\roman*)}
\newlist{enumdef}{enumerate}{1}
\setlist[enumdef]{label=\upshape(\roman*),ref=\thedefn~(\roman*)}
\crefalias{enumdefi}{defn}

\makeatletter
\newcounter{subcreftmpcnt} %
\newcommand\romansubformat[1]{(\roman{#1})} %adapt ....
\def\subcref{\@ifstar\@@subcref\@subcref}
\newcommand\@subcref[2][\romansubformat]{%
	\ifcsname r@#2@cref\endcsname
	\cref@getcounter {#2}{\mylabel}%
	\setcounter{subcreftmpcnt}{\mylabel}%
	\hyperref[#2]{\romansubformat{subcreftmpcnt}}%
	\else ?? \fi}   
\newcommand\@@subcref[2][\romansubformat]{%
	\ifcsname r@#2@cref\endcsname
	\cref@getcounter {#2}{\mylabel}%
	\setcounter{subcreftmpcnt}{\mylabel}%
	\romansubformat{subcreftmpcnt}%
	\else ?? \fi}   
\makeatother

\makeatletter
\DeclareRobustCommand{\crefnosort}[1]{%
	\begingroup\@cref@sortfalse\cref{#1}\endgroup
}
\makeatother

\newcommand{\cD}{{\mathcal D}}\newcommand{\cE}{{\mathcal E}}\newcommand{\cF}{{\mathcal F}}
\newcommand{\cH}{{\mathcal H}}

\newcommand{\cM}{{\mathcal M}}
\newcommand{\cP}{{\mathcal P}}\newcommand{\cQ}{{\mathcal Q}}\newcommand{\cR}{{\mathcal R}}

\newcommand{\cV}{{\mathcal V}}\newcommand{\cX}{{\mathcal X}}

\newcommand{\fP}{{\mathfrak P}}

\newcommand{\fh}{{\mathfrak h}}

\newcommand{\fr}{{\mathfrak r}}

\newcommand{\fv}{{\mathfrak v}}\newcommand{\fw}{{\mathfrak w}}

\newcommand{\BC}{{\mathbb C}}

\newcommand{\BN}{{\mathbb N}}
\newcommand{\BR}{{\mathbb R}}

\usepackage{dsfont} %can be called with option sans for dsletters without serifs

\newcommand{\DSE}{{\mathds E}}

\newcommand{\DSP}{{\mathds P}}

\newcommand{\dsone}{{\mathds 1}}

\usepackage{mathrsfs}

\newcommand{\sD}{{\mathscr D}}
\newcommand{\sI}{{\mathscr I}}

\newcommand{\sfc}{{\mathsf c}}
\newcommand{\sfd}{{\mathsf d}}\newcommand{\sfe}{{\mathsf e}}\newcommand{\sff}{{\mathsf f}}

\newcommand{\sfs}{{\mathsf s}}

\newcommand{\IN}{\BN}\newcommand{\IR}{\BR}\newcommand{\IC}{\BC}
\newcommand{\N}{\BN}\newcommand{\R}{\BR}\newcommand{\C}{\BC}

\newcommand{\PP}{\DSP}\newcommand{\EE}{\DSE}

\newcommand{\hs}{\fh}\newcommand{\HS}{\cH}\newcommand{\VS}{\cV}

\newcommand{\eps}{\varepsilon}\newcommand{\ph}{\varphi}

\renewcommand{\i}{\sfi}\newcommand{\Id}{\dsone} \renewcommand{\d}{\sfd}

\newcommand{\ran}{\operatorname{ran}}\newcommand{\res}[1]{\!\upharpoonright_{#1}}

\newcommand{\supp}{\operatorname{supp}}\DeclareMathOperator*{\essinf}{ess\,inf}

\newcommand{\spd}{\sigma_{\sfd}}\newcommand{\spess}{\sigma_{\sfe\sfs\sfs}}

\DeclareMathOperator*{\slim}{s-lim}

\newcommand{\wt}[1]{\widetilde{#1}}\renewcommand{\bar}[1]{\overline{#1}}

\DeclareFontFamily{U}{mathx}{\hyphenchar\font45}
\DeclareFontShape{U}{mathx}{m}{n}{
	<5> <6> <7> <8> <9> <10>
	<10.95> <12> <14.4> <17.28> <20.74> <24.88>
	mathx10
}{}
\DeclareSymbolFont{mathx}{U}{mathx}{m}{n}
\DeclareFontSubstitution{U}{mathx}{m}{n}
\DeclareMathAccent{\widecheck}{0}{mathx}{"71}
\DeclareMathAccent{\wideparen}{0}{mathx}{"75}

\DeclareFontFamily{OMX}{MnSymbolE}{}
\DeclareFontShape{OMX}{MnSymbolE}{m}{n}{
	<-6>  MnSymbolE5
	<6-7>  MnSymbolE6
	<7-8>  MnSymbolE7
	<8-9>  MnSymbolE8
	<9-10> MnSymbolE9
	<10-12> MnSymbolE10
	<12->   MnSymbolE12}{}
\DeclareSymbolFont{mnlargesymbols}{OMX}{MnSymbolE}{m}{n}
\SetSymbolFont{mnlargesymbols}{bold}{OMX}{MnSymbolE}{b}{n}
\DeclareMathDelimiter{\llangle}{\mathopen}{mnlargesymbols}{'164}{mnlargesymbols}{'164}
\DeclareMathDelimiter{\rrangle}{\mathclose}{mnlargesymbols}{'171}{mnlargesymbols}{'171}
\DeclareMathDelimiter{\lsem}{\mathopen}{mnlargesymbols}{'102}{mnlargesymbols}{'102}
\DeclareMathDelimiter{\rsem}{\mathclose}{mnlargesymbols}{'107}{mnlargesymbols}{'107}
\DeclareMathDelimiter{\langlebar}{\mathopen}{mnlargesymbols}{'152}{mnlargesymbols}{'152}
\DeclareMathDelimiter{\ranglebar}{\mathclose}{mnlargesymbols}{'157}{mnlargesymbols}{'157}
\DeclareMathDelimiter{\lWavy}{\mathopen}{mnlargesymbols}{'137}{mnlargesymbols}{'137}
\DeclareMathDelimiter{\rWavy}{\mathopen}{mnlargesymbols}{'137}{mnlargesymbols}{'137}

\newcommand{\FGamma}{\Gamma}
\newcommand{\FS}{\cF}\newcommand{\dG}{\sfd\FGamma}\newcommand{\ad}{a^\dagger}

\def\titlename{\scshape  FKN Formula and Ground State Energy for the Spin Boson Model with External Magnetic Field}
\title{\LARGE\scshape  FKN Formula and Ground State Energy\\ for the Spin Boson Model\\ with External Magnetic Field}
\newcommand{\shortauthors}{D. Hasler, B. Hinrichs, O. Siebert}

\author{David Hasler\thanks{\texttt{david.hasler@uni-jena.de}}}
\author{Benjamin Hinrichs\thanks{\texttt{benjamin.hinrichs@uni-jena.de}}}
\affil{Friedrich-Schiller-Universität Jena\newline{\small Department of Mathematics\\  Ernst-Abbe-Platz 2\\  07743 Jena\\ Germany}\vspace*{.5em}}
\author{Oliver Siebert\thanks{\texttt{oliver.siebert@uni-jena.de}}\thanks{Present Affiliation: Eberhard-Karls-Universität Tübingen, Institute of Mathematics, An der Morgenstelle 10, D-72076 Tübingen, Germany}}
\affil{\'Ecole polytechnique f\'ed\'erale de Lausanne\newline{\small Institute of Mathematics\\ Station 8\\ CH-1015 Lausanne\\ Switzerland}
}

\newcommand{\one}{{\mathds{1}}}
\renewcommand{\i}{{\mathsf{i}}}

\newcommand{\Htm}{\widetilde{H}}
\newcommand{\Ht}{\widetilde{H}}
\newcommand{\HE}{{\mathcal E}}
\newcommand{\Od}{{\Omega_\downarrow}}
\newcommand{\QE}{{\cQ_{\mathsf E}}}
\newcommand{\SE}{{\Sigma_{\mathsf E}}}
\newcommand{\pE}{{\phi_{\mathsf E}}}
\newcommand{\pEK}{{\phi_{\mathsf E,K}}}
\newcommand{\mE}{{\mu_{\mathsf E}}}
\newcommand{\oE}{{\omega_{\mathsf E}}}
\newcommand{\EEE}{\EE_{\mathsf E}}
\newcommand{\Df}{\sD_\sff}

\newcommand{\s}{{\mathsf s}}
\newcommand{\os}{\otimes_\s}

\newcommand{\bn}[1]{\llangle{#1}\rrangle}
\renewcommand{\spd}{\sigma_{\sf disc}}

\newcommand{\chr}[1]{\mathbf{1}_{#1}}

\begin{document}
	
	\maketitle\thispagestyle{empty}\vspace*{-2em}
	\begin{abstract}\noindent
		We consider the spin boson model with external magnetic field. We prove a path integral formula for the 
	heat kernel, known as Feynman-Kac-Nelson (FKN) formula. We use this path integral representation to express the  
	ground state energy as a stochastic integral.  Based on this connection,
	we determine the expansion coefficients of the ground state energy with respect to the magnetic field strength and express them in terms of correlation functions of a continuous Ising model. 
	From a recently proven correlation inequality, we can then deduce that the second order derivative is finite.
	As an application, we show existence of ground states in infrared-singular situations.
	\end{abstract}

\section{Introduction}

In this paper, we study  the spin boson model with external magnetic field.  This model describes the interaction of a two-level quantum mechanical system with a boson  field in presence of a constant external magnetic field. We derive a Feynman-Kac-Nelson (FKN) formula for this model, which  relates  expectation values of the semigroup generated by  the Hamilton operator  to the expectation value of a Poisson-driven jump process and a  Gaussian random process   indexed by a real Hilbert space obtained by an Euclidean extension of the dispersion relation of the bosons.  
Especially, when calculating  expectation values with respect  to the ground state of the free Hamiltonian, one can  explicitly integrate out the boson  field  and obtain  expectations 
only with respect to  the jump process. This allows us to express the ground state energy and its derivatives in terms of correlation functions of a continuous Ising model, provided a gap assumption is satisfied.
As an application, we show  the existence of ground states for the spin boson model  in the case of massless bosons for infrared singular interactions,
using a recent correlation bound and a regularization procedure.

The history of FKN-type theorems 
dates back to   the work of Feynman and Kac \cite{Feynman.1954,Kac.1951}. Such functional integral respresentations were used to study the spectral properties of models in quantum field theory by Nelson \cite{Nelson.1973}. Since then, many authors have used this approach to study models of non-relativistic quantum field theory, see for example \cite{GlimmJaffe.1985,GlimmJaffe.1987,Spohn.1987,FeffermanFroehlichGraf.1997,Hiroshima.1997,BetzHiroshimaLorincziMinlosSpohn.2002,BetzSpohn.2005,HiroshimaLorinczi.2008,BetzHiroshima.2009} and references therein.
The spin boson model without an external magnetic field has been investigated using this approach in \cite{SpohnDuemcke.1985,FannesNachtergaele.1988,Abdessalam.2011} and recently in \cite{HirokawaHiroshimaLorinczi.2014}. In \cite{Spohn.1989} path measures for the spin boson model with magnetic field were studied  by means of Gibbs measures.
In this paper, we extend the FKN formula for the spin boson model   to external magnetic fields.

%\bigskip\noindent
This paper is structured as follows.
\Cref{sec:modelandres} is devoted to the definition of the spin boson model and the statement of our main results. We start out with a rigorous definition of the spin boson Hamiltonian with external magnetic field as a selfadjoint lower-semibounded operator in \cref{subsec:SBmodel}. In \cref{subsec:FKNformula}, we then describe its probabilistic description through the FKN formula stated in \cref{thm:feynmankac} and reduce the degrees of freedom to study expectation values with respect to the ground state of the free operator as expectation values of a continuous Ising model in \cref{cor:intout}. In \cref{subsec:GSenergy}, we then use the well-known connection between expectation values of the semigroup and the ground state energy to express the derivatives of the ground state energy with respect to the magnetic field strength as correlation functions of this   continuous Ising model, under the assumption of massive bosons.
The proofs of the results presented in \cref{sec:modelandres} are given in \cref{sec:proof}.

In \cref{sec:suszbound}, we then apply our results and prove  \cref{thm:gsexists}. Explicitly, we use the recent result from \cite{HaslerHinrichsSiebert.2021a} to prove the existence of ground states of the spin boson Hamiltonian with vanishing external magnetic field. Our proof especially includes the case of massless bosons with  infrared-singular coupling. 

The article is accompanied by a series of appendices.  In \cref{app:integrals,app:faa,app:FQ}, we present some essential technical requirements for our proofs, including standard Fock space properties in \cref{app:Fockspace} and a construction of the so-called $\cQ$-space in \cref{appsec:Q}. In \cref{app:massive}, we give a proof for the existence of ground states at arbitrary external magnetic field in the case of massive bosons, a case which to our knowledge is not covered in the literature.

\subsubsection*{\textsc{ General Notation}}

{\em $L^2$-spaces:} For a measure space $(\cM,d\mu)$ and a real or complex Hilbert space $\hs$, we denote by $L^2(\cM,d\mu;\hs)$ the real or complex Hilbert space of  square-integrable $\hs$-valued measurable functions on $\cM$, respectively. If $\hs=\IC$, we write for simplicity $L^2(\cM,d\mu)=L^2(\cM,d\mu;\IC)$. Further, we assume $\IR^d$ for any $d\in\IN$ to be equipped with the Lebesgue measure without further mention.

\bigskip \noindent
{\em Characteristic functions:}
For $A\subset X$, we define the function $\chr{A}:X\to\IR$ with $\chr A(x)=1$, if $x\in A$, and $\chr A(x)=0$, if $x\notin A$.

\section{Model and Results}\label{sec:modelandres}

%In this \lcnamecref{sec:modelandres}, we define the spin boson model with external magnetic field and present our main results.

\subsection{Spin Boson Model with External Magnetic Field}\label{subsec:SBmodel}

In this \lcnamecref{subsec:SBmodel}, we give a precise definition of the spin boson Hamiltonian with external magnetic field and prove that it defines a selfadjoint lower-semibounded operator.

%\subsubsection*{\textsc{Fock Spaces}}
%\bigskip\noindent
Let us recall the standard Fock space construction from the Hilbert space perspective. Textbook expositions on the topic can, for example, be found in \cite{ReedSimon.1975,Parthasarathy.1992,BratteliRobinson.1996,Arai.2018}. 

Throughout, we assume $\hs$ to be a complex Hilbert space.
Then, we define the bosonic Fock space over $\hs$ as
\begin{equation}
	\FS(\hs) = \IC\oplus \bigoplus_{n\in\IN} \otimes^n_{\mathrm{s}}\hs,
\end{equation}
where $\otimes^n_{\mathrm{s}}\hs$ denotes the $n$-fold symmetric tensor product of Hilbert spaces. 
We write Fock space vectors as sequences $\psi=\left(\psi^{(n)}\right)_{n\in\IN_0}$ with $\psi^{(0)}\in\IC$ and $\psi^{(n)}\in\otimes^n_{\mathrm{sym}}\hs$. Especially, we define the Fock space vacuum $\Omega=(1,0,0,\ldots)$.

For a self-adjoint operator $A$, let the (differential) second quantization operator $\dG(A)$ on $\FS(\hs)$ be the operator
\begin{equation}
	\dG(A) = \bar{0\oplus \bigoplus_{n\in\IN}\sum_{k=1}^{n}(\otimes^{k-1}\Id)\otimes A \otimes (\otimes^{n-k-1}\Id)},
\end{equation}
where $\bar{(\cdot)}$ denotes the operator closure.
Next, if $\fv$ is another complex Hilbert space and $B:\hs\to\fv$ is a contraction operator (i.e., $\|B\|\le 1$), the second quantization operator $\Gamma(B):\FS(\hs)\to\FS(\fv)$ is given as
\begin{equation}\label{def:G}
	\Gamma(B) = 1\oplus\bigoplus_{n\in\IN}\otimes^nB.
\end{equation}
Furthermore, for $f\in\hs$, we define the creation and annihilation operators $a^\dag(f)$ and $a(f)$ as the closed linear operators acting on pure tensors as
\begin{equation}\label{def:creann}
	\begin{aligned}
		& a(f)g_1\os\cdots\os g_n = \frac{1}{\sqrt{n}}\sum_{k=1}^n\braket{f,g_k}g_1\os\cdots\widehat{g_k}\cdots\os g_n,\\
		& a^\dag(f)g_1\os\cdots \os g_n = \sqrt{n+1}f\os g_1\os\cdots\os g_n,
	\end{aligned}
\end{equation}
where $\os$ denotes the symmetric tensor product and $\ \widehat{\cdot}\ $ in the first line means that the corresponding entry is omitted. Note that $\ad(f)$ is the adjoint of $a(f)$. We introduce the field operator as
\begin{equation}\label{def:field}
	\ph (f) = \frac{1}{\sqrt 2} \bar{a(f)+a^\dag (f)}.
\end{equation}
In \cref{app:Fockspace}, we provide a variety of well-known properties of the operators defined above, which will be used throughout this article. From now on, we also write $\FS=\FS(L^2(\IR^d))$.

%\bigskip\noindent
To define the spin boson Hamiltonian with external magnetic field, let $\sigma_x,\sigma_y,\sigma_z$ denote the $2\times 2$-Pauli matrices
\begin{equation}
	\sigma_x = \begin{pmatrix} 0 & 1 \\ 1 & 0\end{pmatrix},
	\qquad 
	\sigma_y = \begin{pmatrix} 0 & -\i \\ \i & 0\end{pmatrix}
	\qquad
	\sigma_z = \begin{pmatrix} 1 & 0 \\ 0 & -1\end{pmatrix}.
\end{equation}
We consider the Hamilton operator
\begin{equation}\label{def:H}
	H(\lambda,\mu) = \sigma_z\otimes\Id + \Id\otimes \dG(\omega) + \sigma_x\otimes(\lambda \ph(v)+\mu\Id)
	\qquad\mbox{acting on}
	\ \HS=\IC^2\otimes \FS.
\end{equation}
%
%\bigskip\noindent
%We now want t
To prove that the expression \cref{def:H} defines a selfadjoint lower-semibounded operator, we need the following assumptions.
\begin{hyp}\label{hyp:sbmin}\ 
	\begin{enumhyp}
		\item $\omega:\IR^d\to[0,\infty)$ is measurable and has positive values almost everywhere.
		\item $v\in L^2(\IR^d)$ satisfies $\omega^{-1/2}v\in L^2(\IR^d)$.
	\end{enumhyp}
\end{hyp}
\begin{lem}\label{lem:sbmin}
	Assume \cref{hyp:sbmin} holds. Then
	the operator $H(\lambda,\mu)$ given by \cref{def:H} is selfadjoint and lower-semibounded on the domain $\cD(\Id\otimes \dG(\omega))$ for all values of $\lambda,\mu\in\IR$.
\end{lem}
\begin{proof}
	As a sum of strongly commuting selfadjoint and lower-semibounded operators $H(0,0)$ is selfadjoint and lower-semibounded, cf. \cref{lem:Fockprops.pos}. Further, by \cref{lem:Fockprops.standardbound}, with $A=\omega$, and the boundedness of $\sigma_x$, the operator $\sigma_x\otimes (\lambda \ph(v)+\mu\Id)$ is infinitesimally bounded with respect to $H(0,0)$. Hence, the statement follows from the Kato-Rellich theorem (\cite[Theorem X.12]{ReedSimon.1975}).
\end{proof}

\subsection{Feynman-Kac-Nelson Formula}\label{subsec:FKNformula}

In this \lcnamecref{subsec:FKNformula}, we move to a probabilistic description of the spin boson model. Except for \cref{lem:jprops}, all statements are proved in \cref{sec:FKNproof}.

%\bigskip\noindent
The spin part can be described by a jump process, which we construct here explicitly.
To that end, let $(N_t)_{t\ge 0}$ be a Poisson process with unit intensity, i.e., a stochastic process with state space $\IN_0$, stationary  independent increments,  and satisfying
\begin{equation}\label{eq:poisson}
	\PP[N_t=k] = e^{-t}\frac{t^k}{k!} \qquad\mbox{for}\ k\in\IN_0,\ t\ge 0,
\end{equation}
realized on some measurable  space $\Omega$. We refer the reader to \cite{Billingsley.1999}  for a concrete realization of $\Omega$.
Moreover, we can choose $\Omega$ such  that $N_t(\omega)$ is  right-continuous for all $\omega \in \Omega$, see for example \cite[Section 23]{Billingsley.2012}.
Further, let $B $ be a Bernoulli random variable with $\PP[B=1]=\PP[B=-1]=\frac 12$,  which we realize on the space  $\{-1, 1\}$.
Then, we define the jump process $( \widetilde{X}_t)_{t\ge 0}$ on the product space  $\Omega \times \{-1 , 1\}$ (equipped with the product measure) by 
\begin{equation}\label{def:X}
	\widetilde{X}_t(\omega,b) = B(b)(-1)^{N_t(\omega)} , \quad (\omega, b) \in \Omega \times \{ - 1 , 1 \} . 
\end{equation}
To fix a suitable  measure space to  work with, we use  the law of the   process  $(\widetilde{X}_t)_{t \geq 0}$.
That is, we  realize  the stochastic process     on the space  
\begin{equation}
	\sD = \{x:[0,\infty)\to\{\pm 1\} :  x\ \mbox{is right-continuous}\}   , 
\end{equation}
where  we equip $\sD$ with the $\sigma$-algebra generated by the projections $\pi_t(x)=x_t$,  $t \geq 0$. 
The measure, $\mu_X$,    on $\sD$  is  then   given by the pushforward with respect to the map 
\begin{align*} 
 \widetilde{X} : \ %&
   \Omega  \times \{ -1 , 1 \}   \to  \sD,\qquad %\\
  %&
    (\omega, b)  \mapsto ( t \mapsto  \widetilde{X}_t(\omega, b)   )  ,
\end{align*} 
for which it is straightforward to see that it is measurable. 
We define the process $X_t(x) = x_t$ for $x \in \sD$, $t \geq 0$. It follows by construction that the stochastic processes  $X_t$ and   $\widetilde{X}_t$ are equivalent, in the sense that they have the same finite-dimensional distributions. 
 For random variables $Y$ on the measure space $(\sD,\mu_X)$, we define 
\[\EE_X[Y] = \int_{\sD} Y d\mu_X.\]
%\textcolor{red}{For more details on the properties of the so-called Skorokhod space $\sD$, we refer to \cite{Billingsley.1999}. }
We note that by the construction \cref{eq:poisson}, the paths of $X$ $\mu_X$-almost surely have only finitely many jumps in any compact interval. We denote the set of all such paths by $\Df$. The property $\mu_X(\Df)=1$ can alternatively also be deferred from the theory of continuous-time Markov processes, cf. \cite{Resnick.1992,Liggett.2010}.

%\bigskip\noindent
We now want to give a probabilistic description of the bosonic field.
To that end, we define the Euclidean dispersion relation $\oE:\IR^{d+1}\to[0,\infty)$ as $\oE(k,t)=\omega^2(k)+t^2$ and the Hilbert space of the Euclidean field as
\begin{equation}
	\HE = L^2(\IR^{d+1},\oE^{-1}(k,t)d(k,t)).
\end{equation}
Let $\pE$ be the Gaussian random variable indexed by the real Hilbert space
\begin{equation}
	\cR=\{f\in\HE: f(k,t)=\bar{f(-k,-t)}\}
\end{equation}
on the (up to isomorphisms unique) probability space $(\QE,\SE,\mE)$ and denote expectation values w.r.t. $\mE$ as $\EEE$. For the convencience of the reader, we have described a possible explicit construction in \cref{appsec:Q}.
We note that the complexification $\cR_\IC$ is unitarily equivalent to $\cE$, by the map $(f,g)\mapsto f+\i g$, and hence $\FS(\cE)$ and $L^2(\QE)$ are unitarily equivalent, by \cref{prop:Q}.

For $t\in\IR$, we define
\begin{equation}\label{def:jt}
	j_tf(k,s) = \frac{e^{-\i ts}}{\sqrt \pi}\omega^{1/2}(k)f(k).
\end{equation}

\begin{lem}\label{lem:jprops}\ 
	\begin{enumlem}
		\item\label{jpart:isom} \cref{def:jt} defines an isometry
		$j_t:L^2(\IR^d)\to\HE$ for any $t\in\IR$.
		\item\label{jpart:real} If, for almost all $k\in\IR^d$, $f\in L^2(\IR^d)$ satisfies $f(k)=\bar{f(-k)}$ and $\omega(k)=\omega(-k)$, then $j_tf\in\cR$.
		\item\label{jpart:exp} $j_s^*j_t = e^{-|t-s|\omega}$ for all $s,t\in\IR$.
	\end{enumlem}
\end{lem}
\begin{proof}
	The statements follow by the direct calculation
	\[ \braket{j_sf,j_tg}_\HE = \int_{\IR^d}\bar{f(k)}g(k)\int_{\IR}e^{-\i(t-s)\tau}\frac{\omega(k)}{\omega^2(k)+\tau^2}\frac{d\tau}{\pi}dk = \int_{\IR^d}\bar{f(k)}e^{-|t-s|\omega(k)}g(k)dk.\qedhere \]
\end{proof}
\begin{rem}
	In the literature \cref{def:jt} is often defined via the Fourier transform $\widecheck{j_tf} = \delta_t\otimes \widecheck f$.
\end{rem}

We set
\begin{equation}\label{def:It}
	\wt{I}_t:\FS\to L^2(\QE), \qquad \psi\mapsto \Theta_{\cR}\Gamma(j_t)\psi,
\end{equation}
where $\Theta_{\cR}$ denotes the Wiener-It\^o-Segal isomorphism introduced in  \cref{prop:Q} and $\Gamma$ is the second quantization of the contraction operator $j_t$, as defined in \cref{def:G}.
Further, we define the isometry $\iota :  \IC^2 \to  L^2( \{ \pm 1 \}  ,  \mu_{1/2} )$, with $\mu_{1/2}(\{ s \} ) = \frac{1}{2}$  for $s \in \{\pm1\}$,  by 
\begin{equation}\label{def:iota}
	(\iota \alpha)(+1)=\sqrt 2 \alpha_1 \text{ and }  (\iota \alpha)(-1)=\sqrt 2 \alpha_2 ,
\end{equation}
where $\alpha_i$ denotes the $i$--th entry of the vector $\alpha\in\IC^2$.
We define the map  $I_t :=  \iota \otimes \wt{ I}_t$, where 
$$I_t  : \HS = \C^2 \otimes \FS \to   L^2( \{ \pm 1 \}  ,  \mu_{1/2} ) \otimes L^2(\QE) \cong  L^2( \{ \pm 1 \} ,  \mu_{1/2} ;  L^2(\QE)   ) . $$
%
%, 
To formulate the Feynman-Kac-Nelson (FKN) formula, it will be suitable to work with the following transformed Hamilton operator, which is unitary equivalent to $H(\lambda,\mu)$ up to a constant multiple of the identity.
Explicitly,
we apply the unitary
\begin{equation}\label{def:unitary}
	U = e^{\i\frac{\pi}{4}\sigma_y} = \frac{1}{\sqrt 2}\begin{pmatrix}1 & 1\\ -1 & 1\end{pmatrix}
\end{equation}
and define the transformed Hamilton operator
\begin{equation}\label{def:Ht}
	\Ht(\lambda,\mu) = \Id + (U\otimes\Id)H(\lambda,\mu)(U\otimes\Id)^* = (\Id-\sigma_x)\otimes\Id + \Id\otimes\dG(\omega)+\sigma_z\otimes(\lambda\ph(v)+\mu\Id),
\end{equation}
where we used $U\sigma_zU^*=-\sigma_x$ and $U\sigma_x U^*=\sigma_z$.

Our result holds under the following assumptions.
\begin{hyp}\label{hyp:feynmankac} Assume \cref{hyp:sbmin} and the following:
	\begin{enumhyp}
		\item $\omega(k)=\omega(-k)$ for almost all $k\in\IR^d$.
		\item $v$ has real Fourier transform, i.e., $v(k)=\bar{v(-k)}$ for almost all $k\in\IR^d$.
	\end{enumhyp}
\end{hyp}
%Note that by \cref{lem:Fockprops} and the Kato-Rellich theorem, \cref{def:H} defines a self-adjoint lower bounded operator $H(\lambda,\mu)$ for any choice of $\lambda,\mu\in\IR$ if \cref{hyp:feynmankac} holds.
%
%For the probabilistic description of the model, it will be convenient to work with the following transformed Hamilton operator, which is unitary equivalent up to a constant multiple of the identity.
%%
%%However,
%%instead of considering $H(\lambda,\mu)$  as defined in \cref{def:H},
%Explicitly,
%we apply the unitary
%\begin{equation}\label{def:unitary}
%	U = e^{\i\frac{\pi}{4}\sigma_y} = \frac{1}{\sqrt 2}\begin{pmatrix}1 & 1\\ -1 & 1\end{pmatrix}
%\end{equation}
%and define the transformed Hamilton operator
%\begin{equation}\label{def:Ht}
%	\Ht(\lambda,\mu) = \Id + (U\otimes\Id)H(\lambda,\mu)(U\otimes\Id)^* = (\Id-\sigma_x)\otimes\Id + \Id\otimes\dG(\omega)+\sigma_z\otimes(\lambda\ph(v)+\mu\Id),
%\end{equation}
%where we used $U\sigma_zU^*=-\sigma_x$ and $U\sigma_x U^*=\sigma_z$.
%
We are now ready to state the FKN formula for the
%The FKN formula for the
spin boson model with external magnetic field.% is now stated in the following \lcnamecref{thm:feynmankac}.
\begin{thm}[FKN Formula]\label{thm:feynmankac}
	Assume \cref{hyp:feynmankac} holds. Then, 
	for all $\Phi,\Psi\in\HS$ and $\lambda,\mu\in\IR$, we have
	\[\Braket{\Phi,e^{-T\Htm(\lambda,\mu)}\Psi} = \EE_X\EEE\left[\bar{I_0\Phi(X_0)}e^{{-}\lambda\int_0^T\pE\left( j_t v\right)X_tdt {-} \mu\int_0^T X_tdt}I_T\Psi(X_T)\right].\]
\end{thm}
We note that the integrability of the right hand side in above \lcnamecref{thm:feynmankac} follows from the identity
\begin{equation}\label{eq:expgauss}
	\EE\left[\exp(Z)\right] = \exp\left(\frac 12\EE[Z^2] \right),
\end{equation}
which holds for any Gaussian random variable $Z$ (see for example \cite[(I.17)]{Simon.1974}). We outline the argument in the \lcnamecref{rem:integrals} below.
\begin{rem}\label{rem:integrals}
%	Let $x\in\Df$.
%	
%	
%	By the standard definition 	\texttt{Add a discussion which explicitly discusses on paths $x_t$}, the Poisson process $N_t,t\ge 0$ almost surely has only finitely many jumps in any compact interval, and hence $X_t$, $t \geq 0$. Therefore, by the definition \cref{def:X}, the paths of $\mathcal{D}_c$ almost surely have only finitely many points of discontinuity.  Further,
%	
	By \cref{def:jt}, the map $[0,T]\to \cE,t\mapsto j_tv$ is strongly continuous. Hence, by \cref{def:gauss}, the map $\IR\to L^2(\QE),t\mapsto \pE(j_tv)$ is  continuous.
	Thus, for $(x_t)_{t\ge 0}\in \Df$, the function 
	 %It follows that for $\mu_X$-almost every element  $(x_t)_{t \geq 0}$ in $\mathcal{D}_c$   the function 
	 $t \mapsto \phi_E(j_t v) x_t$ is a piecewise continuous  	 $L^2(\QE)$-valued function on compact intervals of $[0,\infty)$. Thus, the integral over $t$ exists as an $L^2(\QE)$-valued  Riemann integral $\mu_X$-almost surely. Since Riemann integrals are given as limits of sums, the  measurability with respect to the product measure  $\mu_X \otimes \mE$  follows. 
	 In fact, again fixing $x\in \Df$ and using Fubini's theorem as well as Hölder's inequality, one can prove that the integral $\int_0^T \pE(j_tv)x_tdt$ can also be calculated as Lebesgue-integral evaluated $\mE$-almost everywhere pointwise in $\QE$ with the same result. This is outlined in  \cref{app:integrals}.
	  Furthermore, $\int_0^T \pE(j_tv)x_tdt$ is a Gaussian random variable, since $L^2$-limits of linear combinations of Gaussians are Gaussian.
	 We conclude that the right hand side of the FKN formula is finite, since exponentials of Gaussian random variables  are integrable, cf. \cref{eq:expgauss}.
\end{rem}
We now want to
describe the expectation value of the semigroup associated with $H(\lambda,\mu)$ (cf. \cref{def:H}) with respect to the ground state of the free operator $H(0,0)$, by integrating out the field contribution in the expectation value. To that end, let 
\begin{equation}\label{def:vac}
	\Od = \begin{pmatrix}0\\1\end{pmatrix}\otimes \Omega
\end{equation}
and define
\begin{equation}\label{def:W}
	W(t) = \frac {1}4\int_{\IR^d}|v(k)|^2e^{-|t|\omega(k)}dk.
\end{equation}
\begin{cor}\label{cor:intout}
	Assume \cref{hyp:feynmankac} holds.
	Then, for all $\lambda,\mu\in\IR$, we have
	\[e^{-T}\Braket{\Od,e^{-TH(\lambda,\mu)}\Od} = \EE_X\left[\exp\left(\lambda^2\int_0^T\int_0^TW(t-s)X_tX_sdsdt {-} \mu\int_0^T X_tdt\right)\right].\]
\end{cor}
\begin{rem}\label{rem:Ising}
	For $(x_t)_{t\ge 0}\in \Df$, the functions $(s,t)\mapsto W(t-s)x_tx_s$ and $t\mapsto x_t$ are Riemann-integrable, since $W$ is continuous. Further, the continuity also implies that the expression on the right hand side is uniformly bounded in the paths $x$ and hence the expectation value exists and is finite by the dominated convergence theorem.
\end{rem}
\begin{rem}
	The expectation value on the right hand side can be interpreted as the partition function of a long-range continuous Ising model on $\IR$ with coupling functions $W$. This model can be obtained as a limit of a discrete Ising model with long-range interactions, see \cite{SpohnDuemcke.1985,Spohn.1989,HaslerHinrichsSiebert.2021b}.
%	
%	\texttt{\textcolor{red}{Extend this remark}}
%	The right hand side is well-defined, since $W$ is continuous and the paths of $X$ are c\`adl\`ag, see also  \cref{rem:integrals}.
%	Moreover, it 
%	has the interpretation of  the   partition function of a long-range continuous Ising model  on $\IR$  with coupling function $W$. 
\end{rem}

\subsection{Ground State Energy}
\label{subsec:GSenergy}

We are especially interested in studying the ground state energy of the spin boson model
\begin{equation}\label{def:E}
	E(\lambda,\mu) = \inf\sigma(H(\lambda,\mu)) . 
\end{equation}
 In this \lcnamecref{subsec:GSenergy}, we want to use the FKN formula from the previous section to express derivatives of the ground state energy.

Starting point of this investigation is the following well-known formula, sometimes referred to as Bloch's formula, expressing the ground state energy as expectation value of the semigroup, see for example \cite{Simon.1979}. We verify it in \cref{sec:GSenergyproof} using a positivity argument.
\begin{lem}\label{lem:Bloch}\label{lem:Bloch.1}
	Assume \cref{hyp:sbmin} holds. Then, for all $\lambda,\mu\in\IR$, \[\displaystyle E(\lambda,\mu) = -\lim_{T\to \infty}\frac 1T\ln\Braket{\Od,e^{-TH(\lambda,\mu)}\Od}.\]
\end{lem}
The central statement of this section is that above equation carries over to the derivatives with respect to $\mu$, provided that  the ground state energy of $H(\lambda,\mu)$ is in the discrete spectrum, i.e., $E(\lambda,\mu) \in  \spd(H(\lambda,\mu)$. We note that this spectral assumption has been shown in  \cite[Theorem 1.2]{AraiHirokawa.1995} for $\mu=0$  if  $\essinf_{k\in\IR^d}\omega(k)>0$ and we extend the result to arbitrary choices of $\mu$ in \cref{app:massive}.
\begin{thm}\label{prop:partition}
	Assume \cref{hyp:sbmin} holds. Let $\lambda ,\mu_0 \in \R$ and suppose  $E(\lambda,\mu_0) \in \sigma_{\sf disc}(H(\lambda,\mu_0))$.
	Then, for all  $n \in \IN$, the following derivatives exist and satisfy 
	\begin{align*} 
		\left.	\partial_\mu^n E(\lambda,\mu) \right|_{\mu = \mu_0}  & = \lim_{ T \to \infty} 	\left. -\frac 1T\partial_\mu^n \ln    \Braket{\Od,e^{-TH(\lambda,\mu)}\Od}\right|_{\mu = \mu_0}   . 
	\end{align*} 
\end{thm}
We now want to combine this observation with the FKN formula from \cref{thm:feynmankac}.
To that end, we define
\begin{equation}\label{def:ZT}
	Z_T(\lambda,\mu) = \EE_X\left[\exp\left(\lambda^2\int_0^T\int_0^T W(t-s)X_tX_sdsdt {-} \mu\int_0^TX_tdt\right)\right],
\end{equation}
with $W$ as defined in \cref{def:W} and note that
\begin{equation}\label{eq:intout0}
	Z_{T}(\lambda,\mu) = e^{-T}\Braket{\Od,e^{-TH(\lambda,\mu)}\Od},
\end{equation}
by \cref{cor:intout}.
Thus,  \cref{lem:Bloch} gives
\begin{equation}\label{eq:intout}
	E(\lambda,\mu)= - \lim_{ T \to \infty}\left(\frac 1T\ln Z_{T}(\lambda,\mu) +1 \right).
\end{equation}
We note that the stochastic integral in \cref{eq:intout} was used in \cite{Abdessalam.2011}  to show analyticity of $ \lambda 
\mapsto E(\lambda,0)$ in a neighborhood of zero. 
The next two statements express the derivatives of the ground state energy in terms of a stochastic integral. To that end, for a random variable $Y$ on $(\sD,\mu_X)$, we define  the expectation
\begin{equation}\label{def:exp}
	\bn{   Y    }_{T,\lambda,\mu} =  \frac{1}{Z_T(\lambda,\mu)} \EE_X   \left[ Y  \exp\left(\lambda^2\int_0^T\int_0^T W(t-s)X_tX_sdsdt {-} \mu\int_0^TX_tdt\right)\right] .
\end{equation}
Further,  we denote by  $\cP_n$ the set of all partitions of the set $\{1,\ldots,n\}$ and by $|M|$ the cardinality of a finite set $M$.
\begin{thm}\label{david123}  
	Assume \cref{hyp:feynmankac} holds. Let $\lambda , \mu  \in \R$ and  suppose  $E(\lambda,\mu) \in \spd(H(\lambda,\mu))$.
	Then, for all   $n \in \IN$, the following derivatives exist and satisfy 
	\begin{align*} 
		\partial_\mu^n E(\lambda, \mu) = &  \lim_{ T \to \infty}\frac{1}{T} \sum_{\fP  \in \mathcal{P}_n} (-1)^{|\fP|+n} ( |\fP|  -1)!  \prod_{B \in  \fP  } 
		\left\llangle   \left( \int_0^T X_t dt \right)^{|B|}    \right\rrangle_{T,\lambda,\mu} . 
	\end{align*}
\end{thm}
%\benjamin{Sollen wir diesen Beweis nochmal gemeinsam durchgehen, um mit Sicherheit das richtige Vorzeichen zu haben?}
In addition, we can  express derivatives of  the ground state energy  in terms of  the so-called Ursell functions \cite{Percus.1975} or cumulants.  This allows us to use 
correlation inequalities  to prove bounds on derivatives. In fact, we will use this in \cref{cor:corr}   below to estimate the second derivative with respect to the magnetic field at zero. 
Given  random variables $Y_1,\ldots,Y_n$   on $(\sD,\mu_X)$, we define the Ursell function 
\begin{equation}\label{def:ursell}
	u_n(Y_1,\ldots,Y_n) = \left. \frac{\partial^n}{\partial h_1 \cdots \partial h_n } \ln  \left\llangle   \exp\left( \sum_{j=1}^n h_i Y_i \right) \right\rrangle_{T, \lambda,\mu} \right|_{h_i = 0}    .
\end{equation} 
\begin{cor}\label{david1234} Assume \cref{hyp:feynmankac} holds. Let $\lambda, \mu  \in \R$ and  suppose  $E(\lambda,\mu) \in \spd(H(\lambda,\mu))$.
	Then, for all   $n \in \IN$, the following derivatives exist and satisfy 
	\begin{align*}  
		\partial_\mu^n E(\lambda, \mu )   & = - \lim_{ T \to \infty}  \frac{1}{T}  \int_{0}^T ds_1  \cdots \int_0^T   ds_n  u_n(X_{s_1},\ldots,X_{s_n} ) .
	\end{align*} 
\end{cor}
Next, we  show how  the formulas in  \cref{david123,david1234}, respectively, can be 
used to obtain bounds 
on  derivatives of the ground state energy. For this, we will use the following 
correlation bound of  a continuous long-range Ising model, 
cf. \cref{rem:Ising}. Approximating this model by a discrete Ising model, we proved a bound on these correlation functions in \cite{HaslerHinrichsSiebert.2021b}.
\begin{thm}[\cite{HaslerHinrichsSiebert.2021b}]\label{prop:corbound}
	There exist $\eps>0$ and $C>0$ such that for all $h\in L^1(\IR)$ which are even, continuous and satisfy $\|h\|_{L^1(\IR)}\le \eps$, we have
	\[0\le \limsup_{T\to\infty}   \frac{\displaystyle  \frac 1T\EE_X \left[\left(\int_0^T X_tdt\right)^2\exp\left(\int_0^T\int_0^Th(t-s)X_tX_sdtds\right)\right]}{\displaystyle \EE_X\left[\exp\left(\int_0^T\int_0^Th(t-s)X_tX_sdtds\right)\right]} \le C . \]
\end{thm}
\noindent
As an application  of   \cref{david123,prop:corbound} we obtain 
the following result, giving us a  bound on the second derivative of the ground state energy which is uniform in the size of the spectral gap. Since the proof only demonstrates the application of these theorems, we state it here directly.
\begin{cor}\label{cor:corr}  Let $\nu$ be a  measurable function on $\R^d$  
	satisfying $\nu > 0$ a.e.  and $\nu(-k) = \nu(k) $.  Let $v \in L^2(\R^d)$ have real Fourier transfrom  and $\nu^{-1/2}v\in L^2(\IR^d)$.
	Let  $m > 0$ and $\omega = \sqrt{  \nu^2 + m^2}$.  Then, for every $\lambda \in \R$,
	the function  $ \mu \mapsto E(\lambda,\mu)$ is twice differentiable  in a neighborhood of zero and, choosing $W$ as defined in \cref{def:W},
	\begin{align*} 
		\partial_\mu^2 E(\lambda , 0 ) %= -   \lim_{T \to \infty} T^{-1} \left \llangle  \left( \int_0^T  X_t dt \right)^2 \right \rrangle_{T, \lambda, 0 } ,
		& = -  \lim_{ T \to \infty}\frac{\displaystyle \frac 1T  \EE_X\left[\left(\int_0^T X_tdt\right)^2\exp\left(\lambda^2\int_0^T\int_0^TW(t-s)X_tX_sdtds\right)\right]}{\displaystyle \EE_X\left[\exp\left(\lambda^2\int_0^T\int_0^TW(t-s)X_tX_sdtds\right)\right]} . 
	\end{align*} 
	Further, there exists a $\lambda_\sfc > 0$ such that for all $\lambda  \in  (-  \lambda_\sfc , \lambda_\sfc)$  the second derivative satisfies 
	\[ \limsup_{ m \downarrow 0}|  \partial_\mu^2 E(\lambda , 0 )| < \infty.\]   
\end{cor}
\begin{proof}
	Due to the definition,  we have $\essinf_{k\in\IR^d}\omega(k)\ge m>0$ and hence $E(\lambda,0)\in\spd(H(\lambda,0))$, by \cref{thm:massive}.
	Thus,    \cref{david123} is applicable. %  with $W$ as defined in \cref{def:W}.
	
%	Since $X$ and $-X$ are equivalent stochastic processes ,
%	Now,   observe that 
	Due to the so-called spin-flip-symmetry of the model, i.e., $X$ and $-X$ being equivalent stochastic processes in the sense of their finite-dimensional distributions by the choice of the Bernoulli random variable in \cref{def:X},
	we have
	\[ \left \llangle \int_0^T X_t dt \right\rrangle_{T,\lambda,0} = \left \llangle \int_0^T( -X_t) dt \right\rrangle_{T,\lambda,0} = -\left \llangle \int_0^T X_t dt \right\rrangle_{T,\lambda,0}, \]
	and hence
	$$\left \llangle \int_0^T X_t dt \right\rrangle_{T,\lambda,0} = 0 \qquad\mbox{for any value of}\ T\ge 0. $$
	Thus, by  \cref{david123}
	\[\partial_\mu^2 E(\lambda,0) = -\lim_{T \to \infty}  \frac 1T \left \llangle  \left( \int_0^T  X_t dt \right)^2 \right \rrangle_{T, \lambda, 0 } 
	.\]
	By the definition in \cref{def:W}, the interaction function $W\in L^1(\IR)$ satisfies 
	\[\| W \|_{L^1(\IR)}  = \frac{1}{2} \| \omega^{-1/2} v \|_2^2 \leq \frac{1}{2}  \| \nu^{-1/2} v \|_2^2.\]
	Setting $\lambda_\sfc = (\frac{1}{2} \eps )^{1/2} /\| \nu^{-1/2} v \|_2$ with $\eps$ given  as in \cref{prop:corbound}, we can apply \cref{prop:corbound} with $h=\lambda^2 W$ for all $\lambda\in(-\lambda_\sfc,\lambda_\sfc)$, which proves $|\partial_\mu^2 E(\lambda,0)|\le C$ for all $m>0$. This concludes the proof.
\end{proof}

\section{Proofs}\label{sec:proof}

In this \lcnamecref{sec:proof}, we prove the results presented in \cref{subsec:FKNformula,subsec:GSenergy}

\subsection{The FKN Formula}\label{sec:FKNproof}

We start with the proof of \cref{thm:feynmankac}.
To that end, we first derive a FKN formula for the spin part, which is described by the jump process. For the statement, we recall the definition of $\iota:\IC^2\to L^2({\pm 1},\mu_{1/2})$ in \cref{def:iota}.
\begin{lem}\label{lem:spin}
	Let $n\in\IN$ and $t_1,\ldots,t_n\ge 0$. We set $s_k = \sum\limits_{i=1}^{k}t_k$ for $k=1,\ldots,n$.\\ Then, for all $\alpha,\beta\in\IC^2$ and $f_0,f_1,\ldots,f_n:\{\pm1\}\to\IC$, we have
	\[e^{-s_n}\braket{\alpha,f_0(\sigma_z)e^{t_1\sigma_x}f_1(\sigma_z)e^{t_2\sigma_x}\cdots e^{t_n\sigma_x}f_n(\sigma_z)\beta} = \EE_X\left[\bar{ \iota \alpha({X_0})}f_0(X_0)f_1(X_{s_1})\cdots f_n(X_{s_n}) \iota \beta({X_{s_n})}\right].\]
\end{lem}
\begin{proof}Since any function $f:\{\pm 1\}\to\IC$ is a linear combination of the identity and the constant function $1$, it suffices to consider the case $f_0=f_1=\cdots= f_n= \operatorname{id}$. Further, due to bilinearity, it suffices to choose $\alpha$ and $\beta$ to be arbitrary basis vectors. We hereby use the basis consisting of eigenvectors of $\sigma_x$, i.e.,
	$e_1=\frac1{\sqrt2}(1,1)$ and $e_2=\frac1{\sqrt2}(1,-1)$. Then
	\[ \sigma_x e_1 = e_1,\quad \sigma_x e_2 = -e_2, \quad \sigma_z e_1 = e_2,  \quad\mbox{and}\quad \sigma_z e_2 = e_1  \]
	and hence
	\begin{equation}\label{eq:expsigma}
		\begin{aligned}
			&\braket{e_1,\sigma_z e^{t_1\sigma_x}\cdots e^{t_n\sigma_x}\sigma_z e_1} = \begin{cases} 0 & \mbox{if}\ n\ \mbox{is even},\\ e^{\sum_{j=1}^{n}(-1)^j t_j } & \mbox{if}\ n\ \mbox{is odd}, \end{cases}\\
			&\braket{e_2,\sigma_z e^{t_1\sigma_x}\cdots e^{t_n\sigma_x}\sigma_z e_1} = \begin{cases} e^{ - \sum_{j=1}^{n}(-1)^j t_j     } & \mbox{if}\ n\ \mbox{is even},\\ 0 & \mbox{if}\ n\ \mbox{is odd}, \end{cases}\qquad \\
			&\braket{e_2,\sigma_z e^{t_1\sigma_x}\cdots e^{t_n\sigma_x}\sigma_z e_2} = \begin{cases} 0 & \mbox{if}\ n\ \mbox{is even},\\ e^{- \sum_{j=1}^{n}(-1)^j t_j    } & \mbox{if}\ n\ \mbox{is odd}. \end{cases}
		\end{aligned}
	\end{equation}
	Now, observe that $X_0$, $X_tX_s$ and $X_uX_v$ are independent random variables if $0\le t\le s\le u\le v$, by construction.
	For $i,k\in\IN$, this yields the identities
	\begin{align*}
		&\EE_X[X_{s_i}\cdots X_{s_{i+k}}] = \EE_X[ (X_0)^k ]\EE_X[(X_0X_{s_i})^k]\prod_{j=1}^{k}\EE_X[(X_{s_{i+j-1}}X_{s_{i+j}})^{k-j+1}],\\
		&\EE_X\left[X_{s_i}\cdots X_{s_{i+2k}} \right] = \EE_X[ X_{s_i}X_{s_{i+1}}]\EE_X[ X_{s_i+2}X_{s_{i+3}}] \cdots \EE_X[X_{s_{i+2k-1}}X_{s_{i+2k}}].
	\end{align*}
	If $k$ is odd, then $\EE_X[ (X_0)^k ]=0$, by \cref{def:X} and the definition of $B$.
	 Further, by \cref{eq:poisson}, we find
	\[ \EE_X[X_tX_s] = \sum_{k=0}^{\infty}(\PP[N_{t-s}=2k]-\PP[N_{t-s}=2k+1]) =  e^{-2(t-s)} \qquad\mbox{for}\ 0\le s\le t.  \]
	Hence, setting $s_0=0$, we arrive at
	\begin{equation}\label{eq:expX}
		\EE_X\left[X_{s_j}\cdots X_{s_k} \right] = \begin{cases} 0 & \mbox{if}\ k-j\ \mbox{is even,}\\ e^{-2(t_{j+1}+t_{j+3}+\cdots+t_{k-2}+t_{k})} & \mbox{if}\ k-j\ \mbox{is odd}  \end{cases}
		\qquad\mbox{for}\ 0\le j\le k\le n.
	\end{equation}
	Combining \cref{eq:expX,eq:expsigma}, we have
	\begin{align*}
		&e^{-s_n}\braket{e_1,\sigma_z e^{t_1\sigma_x}\cdots e^{t_n\sigma_x}\sigma_z e_1} = \EE_X\left[X_0X_{s_1}\cdots X_{s_n}\right] , \\
		&e^{-s_n}\braket{e_2,\sigma_z e^{t_1\sigma_x}\cdots e^{t_n\sigma_x}\sigma_z e_1} = \EE_X\left[X_{s_1}\cdots X_{s_n}\right] ,  \\
		&e^{-s_n}\braket{e_2,\sigma_z e^{t_1\sigma_x}\cdots e^{t_n\sigma_x}\sigma_z e_2} = \EE_X\left[X_{s_1}\cdots X_{s_{n-1}}\right].
	\end{align*}
	Observing that $\iota e_1(x)=1$ and $ \iota e_2(x)=x$ for $x=\pm 1$ finishes the proof.
\end{proof}
%\noindent
We now move to proving the FKN formula for the field part. We recall the definition of the isometry $j_t:L^2(\IR^d)\to\cE$ in \cref{def:jt}.
For $I\subset \IR$, let $e_I$ denote the projection onto the closed subspace $\bar{{\rm lin}\{f\in \HE:f\in\operatorname{Ran}(j_t)\ \mbox{for some}\ t\in I\}}$. Further, set $e_t=e_{\{t \}} $ for any $t \in \R$. 
\begin{lem}\label{lem:eprops} Assume $a\le b\le t\le c\le d$. Then
	\begin{enumlem}
		\item\label{epart:simple} $e_t = j_t j_t^*$,
		\item\label{epart:calc} $e_ae_be_c = e_ae_c$,
		\item\label{epart:markov} $e_{[a,b]}e_te_{[c,d]}=e_{[a,b]}e_{[c,d]}$.
	\end{enumlem}
\end{lem}
\begin{proof}\cref{jpart:isom} and the definition of $e_{\{t\}}$ directly imply \subcref{epart:simple}. Further, \subcref{epart:calc} follows from  \cref{jpart:exp} by
	\[e_ae_be_c=j_aj_a^*j_bj_b^*j_cj_c^* = j_ae^{-(b-a)\omega}e^{-(c-b)\omega}j_c^* = j_ae^{-(c-a)\omega}j_c^* = j_aj_a^*j_cj_c^*=e_ae_c.\]
	To prove \subcref{epart:markov}, let $f,g\in \HE$. By the definition, there exist sequences of times $(t_k)_{k\in \IN}\subset [a,b]$ and $(s_m)_{m\in\IN}\subset [c,d]$ and functions $f_k\in \operatorname{Ran}(j_{t_k})$, $g_m\in\operatorname{Ran}(j_{s_m})$ such that
	\[ e_{[a,b]}f = \sum_{k=1}^{\infty}f_k \qquad\mbox{and}\qquad e_{[c,d]}g=\sum_{m=1}^{\infty}g_m.\]
	Furthermore, again by definition  $ \operatorname{Ran}(j_{t}) = \operatorname{Ran}(e_{t})$ for any $t \in \R$.  
	Hence, we can apply \subcref{epart:calc} and obtain
	\[\braket{e_{[a,b]}e_te_{[c,d]}g,f} = \sum_{k,m=1}^{\infty}\braket{e_tg_m,f_k} =  \sum_{k,m=1}^{\infty}\braket{g_m,f_k} = \braket{e_{[a,b]}e_{[c,d]}g,f}.\]
	Since $f$ and $g$ were arbitrary, this proves the statement.
\end{proof}
%\noindent
Now, for $t\in\IR$ and $I\subset \IR$, let
\begin{equation}
	J_t = \Gamma(j_t), \qquad E_t=\Gamma(e_t) \qquad\mbox{and}\qquad E_I = \Gamma(e_I).
\end{equation}
Then the next statement in large parts follows directly from \cref{lem:jprops,lem:eprops,lem:Fockprops}.
\begin{lem}\label{lem:Iprops}Assume $a\le b\le t\le c\le d$ and $I\subset \IR$. Then
	\begin{enumlem}
		\item $E_I$ is the orthogonal projection onto $\bar{{\rm lin}\{f\in \FS(\HE):f\in\operatorname{Ran}(J_t)\ \mbox{for some}\ t\in I\}}$.
		\item $E_t=J_tJ_t^*$,
		\item $E_aE_bE_c=E_aE_c$,
		\item\label{Ipart:markov} $E_{[a,b]}E_tE_{[c,d]}=E_{[a,b]}E_{[c,d]}$.
		\item\label{Ipart:scalar} For all $F\in\operatorname{Ran}(E_{[a,b]}) $ and $G\in\operatorname{Ran}(E_{[c,d]})$, we have $\braket{F,E_tG}=\braket{F,G}$.
		\item\label{Ipart:FKN} $J_s^*J_t = e^{-|t-s|\dG(\omega)}$ for all $s\in\IR$,
		\item\label{Ipart:proj} $J_t\ph(f)J_t^*=E_t\ph(j_tf)E_t = \ph(j_tf)E_t$ for all $f\in L^2(\IR^d)$.
		\item\label{Ipart:proj2} $J_tG(\ph(f))J_t^*=E_t G(\ph(j_tf))E_t = G(\ph(j_tf))E_t$ for all $f\in L^2(\IR^d)$ and bounded measurable functions $G$  on $\R$.
	\end{enumlem}
\end{lem}
\begin{proof} All statements except for \subcref{Ipart:scalar}--\subcref{Ipart:proj} follow trivially from \cref{lem:eprops,lem:Fockprops} and the definitions. \subcref{Ipart:scalar} follows from \subcref{Ipart:markov}, by the simple calculation
	\[ \braket{F,E_tG}=\braket{E_{[a,b]}F,E_tE_{[c,d]}G} = \braket{F,E_{[a,b]}E_tE_{[c,d]}G} = \braket{F,E_{[a,b]}E_{[c,d]}G} = \braket{E_{[a,b]}F,E_{[c,d]}G} = \braket{F,G}.  \]
	Finally, \subcref{Ipart:FKN} and \subcref{Ipart:proj} follow by combining \cref{jpart:exp,lem:Fockprops}. Repeated application of  \subcref{Ipart:proj} shows that  \subcref{Ipart:proj2} holds
	for $G$ a polynomial. That it holds for arbitrary bounded measurable $G$ follows from the measurable functional calculus \cite{ReedSimon.1972}. 
\end{proof}
%\noindent
We can now prove the full FKN formula.
\begin{proof}[\textbf{Proof of \cref{thm:feynmankac}}]
	Throughout this proof, we drop tensor products with the identity in our notation. Further, for the convenience of the reader, we explicitly state in which Hilbert space the inner product is taken.
	
	Let  $\chi_K(x) = \min\{  x , K \} $ if $x \geq 0$, and $\chi_K(x) = \max\{x   , -K\}$ if $x   <  0$.   
	%be the characteristic function of the interval $(-K,K)$
	%and define
 Let 	$\ph_K(v)=\chi_{K}(\ph(v))$, $\pEK(j_tv)=\chi_{K}(\pE(j_tv))$ and $\Ht_K(\lambda,\mu)$ as in \cref{def:Ht} with $\ph$ replaced by $\ph_K$. Since $\Ht_K(\lambda,\mu)$ is lower-semibounded and $\ph_K$ is bounded, 
	we can use the Trotter product formula (cf. \cite[Theorem VIII.31]{ReedSimon.1972}) and \cref{lem:Iprops} Parts  \subcref{Ipart:FKN} and \subcref{Ipart:proj2} (where the exponential is considered on the eigenspaces of $\sigma_x$) to obtain
	\begin{align*}
		e^T\Braket{\Phi,e^{-T\Ht_K(\lambda,\mu)}\Psi}_\HS &= \lim\limits_{N\to\infty}\Braket{\Phi,\left(e^{-\frac TN\dG(\omega)}e^{-\frac TN\sigma_z}e^{-\frac TN \sigma_x\otimes(\lambda\ph_K(v)+\mu)}\right)^n\Psi}_\HS\\
		&= \lim\limits_{N\to\infty}\Braket{\Phi,\prod_{k=1}^N\left( J_{(k-1)\frac TN}^*J_{k\frac TN}e^{-\frac TN\sigma_z}e^{-\frac TN \sigma_x\otimes(\lambda\ph_K(v)+\mu)}\right) \Psi}_\HS\\
		&=\lim\limits_{N\to\infty}\Braket{J_0\Phi,\prod_{k=1}^N\left( J_{k\frac TN}e^{-\frac TN\sigma_z}e^{-\frac TN \sigma_x\otimes(\lambda\ph_K(v)+\mu)}J_{k\frac TN}^*\right) J_T\Psi}_{\IC^2\otimes\FS(\HE)}\\
		& = \lim\limits_{N\to\infty}\Braket{J_0\Phi,\prod_{k=1}^N\left( E_{k\frac TN}e^{-\frac TN\sigma_z}e^{-\frac TN \sigma_x\otimes(\lambda\ph_K(j_{k\frac TN}v)+\mu)}E_{k\frac TN}\right) J_T\Psi}_{\IC^2\otimes\FS(\HE)} . 
	\end{align*}
	Now we make iterated use of \cref{Ipart:scalar}. Explicitly, by \cref{Ipart:proj2}, the vector to the left of any $E_{k\frac TN}$, i.e.,
	\begin{align*}
		&\prod_{j=0}^{k-1} \left( E_{j\frac TN}e^{-\frac TN\sigma_z}e^{-\frac TN \sigma_x\otimes(\lambda\ph_K(j_{k\frac TN}v)+\mu)}E_{j\frac TN}\right)  J_0 \Phi \in \operatorname{Ran}(E_{(k-1)\frac TN}), \\
		& e^{-\frac TN\sigma_z}e^{-\frac TN \sigma_x\otimes(\lambda\ph_K(j_{k\frac TN}v)+\mu)}\prod_{j=0}^{k-1} \left( E_{j\frac TN}e^{-\frac TN\sigma_z}e^{-\frac TN \sigma_x\otimes(\lambda\ph_K(j_{k\frac TN}v)+\mu)}E_{j\frac TN}\right)  J_0 \Phi \in \operatorname{Ran}(E_{k\frac TN})
	\end{align*}
	 is an element of $\operatorname{Ran}(E_{[0,k\frac TN]})$. Equivalently, the vector to the right is an element of $\operatorname{Ran}(E_{[k\frac TN,T]})$. Hence, we can drop all the factors $E_{k\frac TN}$. Then, using \cref{prop:Q,def:It}, we derive
	\begin{align*}
		e^T\Braket{\Phi,e^{-T\Ht_K(\lambda,\mu)}\Psi}_\HS & = \lim\limits_{N\to\infty}\Braket{J_0\Phi,\prod_{k=1}^N\left( e^{-\frac TN\sigma_z}e^{-\frac TN \sigma_x\otimes(\lambda\ph_K(j_{k\frac TN}v)+\mu)}\right) J_T\Psi}_{\IC^2\otimes\FS(\HE)}\\
		&=\lim\limits_{N\to\infty}\Braket{I_0\Phi,\prod_{k=1}^N\left( e^{-\frac TN\sigma_z}e^{-\frac TN \sigma_x\otimes(\lambda\pEK(j_{k\frac TN}v)+\mu)}\right) I_T\Psi}_{\IC^2\otimes L^2(\QE)}.
	\end{align*}
	Hence, we can apply \cref{lem:spin} to obtain
	\begin{align} \label{intinexp} 
		\Braket{\Phi,e^{-T\Ht_K(\lambda,\mu)}\Psi}_\HS &= \lim\limits_{N\to\infty}\EE_X\EEE\left[\bar{I_0\Phi(X_0)}e^{-\frac TN\sum\limits_{k=1}^{N}\left(\lambda\pEK(j_{k\frac TN}v)+\mu\right)X_{k\frac TN}}I_T\Psi(X_T)\right] .
	\end{align}
	Since $\chi_K$ is Lipschitz continuous, it follows that $ t \mapsto \pEK(j_{t} v)$ is an $L^2(\QE)$-valued continuous function.  Thus, the sum in the exponential in  \cref{intinexp}   converges to an   $L^2(\QE)$-valued Riemann
	integral. By possibly going over to a subsequence the Riemann sum  converges $\mu_X \otimes \mE$-almost everywhere. Thus it follows by dominated convergence   that  
	\begin{equation}\label{eq:lastFKN}
		\Braket{\Phi,e^{-T\Ht_K(\lambda,\mu)}\Psi}_\HS = \EE_X\EEE\left[\bar{I_0\Phi(X_0)}e^{{-}\lambda\int_0^T\pEK\left( j_t v\right)X_tdt {-} \mu\int_0^T X_tdt}I_T\Psi(X_T)\right].
	\end{equation}
	(Alternatively, the convergence could also be deduced by estimating  the   expectation.)
	Since $\ph(v)$ is bounded with respect to $\dG(\omega)$ (cf. \cref{lem:Fockprops}), the spectral theorem implies that $\Ht_K(\lambda,\mu)$ converges to $\Ht(\lambda,\mu)$ in the strong resolvent sense and hence the left  hand side of above equation converges to $\Braket{\Phi,e^{-T\Ht(\lambda,\mu)}\Psi}$ as $K\to\infty$. On the other hand, using  that  for   $\mu_X  \otimes \mE $-almost every $(x,q) \in \sD \times \QE$ 
	the  function  $t \mapsto (\pEK(j_tv))(q)x_t$ is Lebesgue integrable,  see Remark  \ref{rem:integrals}, it follows  that $\int_0^T  \pEK(j_tv) x_t dt $ converges to  $\int_0^T  \pE(j_tv) x_t dt $ 
	almost everywhere.
	Hence, the right   hand side of \cref{eq:lastFKN} converges to
		\[\EE_X\EEE\left[\bar{I_0\Phi(X_0)}e^{{-}\lambda\int_0^T\pE\left( j_t v\right)X_tdt {-} \mu\int_0^T X_tdt}I_T\Psi(X_T)\right] \]
	as $K\to\infty$, by the dominated convergence theorem.
	For the majorant,   we use % the pointwise  a.e.  $\mu_X \otimes \mE$ Lebesgue inegrability in $t $,  see Remark  \ref{rem:integrals} and  
	that by Jensen's inequality  
	\begin{align*} 
	\exp( {-}\lambda  \int_0^T\pEK\left( j_t v\right)X_tdt )
	 &  \leq  \frac{1}{T} \int_0^T \exp( {-}\lambda T \pEK\left( j_t v\right)X_t  )  dt \\
	& \leq  \frac{1}{T} \int_0^T  [ \exp( {-}\lambda T \pE \left( j_t v\right) ) + \exp( \lambda T \pE \left( j_t v\right) ) ] dt  , 
	\end{align*} 
	where in the second line we used  $\max\{ e^x , 1 \}  \leq e^x + e^{-x}$. Now the right hand side is integrable over $\QE$-space by \eqref{eq:expgauss}.
	This proves the statement.
	%
	%	The considerations from \cref{rem:integrals} together with the Gaussian decay show, that the Riemann sum converges to the integral in $L^2$-sense, which in turn proves the statement.
	%	
	%	\texttt{\textcolor{red}{Refine below argument using \cref{eq:expgauss}} (continuity of $\exp$)}
	%	
	%	{\tt Corection:} 
	%	Use 
	%	\begin{align*} 
	%	\E | e^{X} - e^{X_n} |  & \leq \E  | e^{X}| 1  - e^{X_n-X} | \leq \E | e^{|X|}  \sum_{m=1}^\infty \frac{| X_n - X |^m}{m!}    \\
	%	& \leq  \E| X_n - X|  | e^{|X|}  \sum_{m=0}^\infty \frac{| X_n - X |^{m}}{(m+1)!}      \\
	%	& \leq   \E| X_n - X|  |e^{|X|} e^{|X_n - X|} \leq \left(  \E | X_n - X|^2    \E   e^{2 |X|+ 2 |X-X_n|}  \right)^{1/2}\\
	%	&  \leq \left(  \E | X_n - X|^2   \right)^{1/2} \left( \E   e^{4 |X|} \E e^{4 |X-X_n|}  \right)^{1/4}\\
	%	\end{align*} 
	%		\begin{align*} 
	%	\E | X - X_n|^2   \leq C^2  
	%	\end{align*} 
	%		\begin{align*} 
	%	\int_0^\infty e^{ a x } \sigma^{-1} e^{- \frac{x^2}{\sigma^2} } dx  &  = 	\int_0^\infty e^{ a \sigma x } e^{-  x^2 } dx    = e^{ a^2 \sigma^2/4}  \int_0^\infty e^{ -  ( x - a \sigma/2)^2  }dx  \\
	%	& \leq   e^{ a^2 \sigma^2/4}  \int e^{- u^2} du 
	%	\end{align*} 
\end{proof}\noindent
We end this section with the
\begin{proof}[\textbf{Proof of \cref{cor:intout}}]
	First, observe that  with $U$ as in  \eqref{def:unitary}, we have
	$(I_t(U^*\otimes\Id)\Od)(x) = 1$ for $x=\pm 1$ and $t\in\IR$ (cf. \cref{def:G,prop:Q}). Hence, \cref{thm:feynmankac} implies
	\[ e^{-T}\Braket{\Od,e^{-TH(\lambda,\mu)}\Od} = \EE_X\left[\EEE\left[\exp\left({-}\lambda\int_0^T\pE\left(j_tv\right) X_t dt\right)\right]\exp\left({-}\mu\int_0^T X_tdt\right)\right].\]
	%	We recall that for any Gaussian random variable $Z$   one has  (see for example \cite[(I.17)]{Simon.1974})
	%	\[ \EE\left[\exp(Z)\right] = \exp\left(\frac 12\EE[Z^2] \right) .\]
	Now, let $x\in\Df$. Then, Fubini's theorem and the identity \cref{eq:expgauss} yield
	\[\EEE\left[\left(\int_0^T\pE\left( j_t v\right)x_tdt\right)^2\right] = \int_0^T\int_0^T \EEE\left[\pE(j_tv)\pE(j_sv)\right]x_tx_sdtds\]
	Now, the definitions of the $\cR$-indexed Gaussian process (cf. \cref{def:gauss}) and $W$ \cref{def:W} imply
	\begin{align*}
		\int_0^T\int_0^T \EEE\left[\pE(j_tv)\pE(j_sv)\right]x_tx_sdtds
		&=\int_0^T\int_0^T \frac 12\braket{j_tv,j_sv}x_tx_sdtds\\&=
		2\int_0^T\int_0^T W(t-s)x_tx_sdtds,
	\end{align*}
	where we used $j_s^*j_t = e^{-|t-s|\omega}$ (cf. \cref{lem:jprops}).
	This proves the statement.
\end{proof}

\subsection{Derivatives of the Ground State Energy}\label{sec:GSenergyproof}

To prove our results on derivatives of the ground state energy, we start out with the proof of \cref{lem:Bloch}.
Let us first state our version of Bloch's formula. For the convenience of the reader, we provide the simple proof. Similar arguments are, for example, used in \cite{LorincziMinlosSpohn.2002,AbdessalamHasler.2012}.
%The version below is from \cite[Theorem C.1]{MatteMoller.2018} and based on \cite{LorincziMinlosSpohn.2002,AbdessalamHasler.2012}.

Since the notion of positivity is essential therein, we recall it here for the convenience of the reader. For an arbitrary measure space $(\cM,\mu)$, we call a function $f\in L^2(\cM,d\mu)$ {\em (strictly) positive} if it satisfies $f(x)\ge 0$ ($f(x)>0$) for almost all $x\in\cM$. If $A$ is a bounded operator on $L^2(\cM,d\mu)$, we say $A$ is {\em positivity preserving (improving)} if $Af$ is (strictly) positive for all non-zero positive $f\in L^2(\cM,d\mu)$. 

\begin{lem}\label{applem:Bloch2}
	Let $(\cM,\mu)$ be a probability space and let $H$ be a selfadjoint and lower-semibounded operator on $L^2(\cM,d\mu)$. If $e^{-TH}$ is positivity preserving for all $T\ge 0$ and $f\in L^2(\cM,d\mu)$ is striclty positive, then
	\[\inf \sigma(H) = -\lim_{T\to\infty}  \frac 1T\ln\Braket{f,e^{-TH}f}.\]
\end{lem}
\begin{proof}
	First, we note that by the spectral theorem
	\begin{equation}\label{eq:AH21}
		E_g := \inf \supp \nu_g = -\lim_{ T \to \infty}\frac 1T\ln \braket{g,e^{-TH}g} \qquad\mbox{for all}\ g\in L^2(\cM,d\mu),
	\end{equation}
	where $\nu_g$ denotes the spectral measure of $H$ associated with $g$. This easily follows from the inequality \[\nu_g([E_g,E_g+\eps))e^{-T(E_g+\eps)} \le \int_{E_g}^{\infty} e^{-T x}d\nu_g(x) = \braket{g,e^{-TH}g} \le e^{-TE_g} \qquad\mbox{for all}\ T\in\IR,\ \eps>0. \]
	Now, for $h\in L^2(\cM,d\mu)$ satisfying
	\begin{equation}  \label{ineqpos} 
		c_1 f  \leq h \leq c_2 f   \quad  \mbox{for some}  \  c_1 , c_2 > 0 , 
	\end{equation} 
	it follows from $e^{-TH}$ being positivity preserving that
	\[ c_1^2\braket{f,e^{-TH}f} \le \braket{h,e^{-TH}h} \le c_2^2 \braket{f,e^{-TH}f}. \]
	Combined with \cref{eq:AH21}, it follows that $\inf \sigma(H)\le E_f=E_h$.
	Since $f$ is strictly positive, the linear span of the set $\cX$ of functions satisfying \cref{ineqpos} is dense in $L^2(\cM,d\mu)$. It follows that $E_f = \sigma(H)$, since otherwise $E_f > \inf \sigma(H)$ and $\chi_{(-\infty, E_f)}(H) L^2(\cM , d\mu)$
	would contain a nonzero vector which is orthogonal to $\cX$. Thus, the statement follows from \cref{eq:AH21}.
%	
%	Let $\mathcal{X}$ denote the set of 
%	functions of the form   \eqref{ineqpos}. The linear span of $\mathcal{X}$ is dense in $L^2(M,d\mu)$. It follows that $E_f = \sigma(H)$, since otherwise $E_f > \inf \sigma(H)$ and $\chi_{(-\infty, E_f)}(H) L^2(M , \mu)$
%	would contain a nonzero vector which is orthogonal to $\mathcal{X}$.
\end{proof}
We now prove that the transformed spin boson Hamiltonian from \cref{def:Ht} is positivity improving in an appropriate $L^2$-representation.
\begin{lem} \label{lem:posimp}
	Let $\vartheta$ be the natural isomorphism $\HS=\IC^2\otimes \FS\to L^2(\{1,2\}\times \cQ_{L^2 (\IR^d;\IR)})$ which is determined by
	$ \alpha \otimes \psi \mapsto \left((i,x)\mapsto \alpha_i (\Theta_{L^2 (\IR^d;\IR)} \psi) (x) \right),  $
	where $\Theta_{L^2 (\IR^d;\IR)}:\FS\to L^2(\cQ_{L^2 (\IR^d;\IR)})$ is the unitary from
	\cref{prop:Q}. Then the operator
		$ \vartheta e^{-T\Ht(\lambda,\mu)} \vartheta^* 
	$
	 is positivity improving for all $T>0$.
	 \end{lem} 
	 
	 \begin{proof} 
First observe that 
	\begin{align}  \label{eq:Htilde} \vartheta e^{-T\Ht(\lambda,\mu)} \vartheta^*  =  e^{ - T \vartheta \tilde{H}(\lambda,\mu)\vartheta^*} 
		=  e^{ - T (  \vartheta \widetilde{H}(0,\mu)\vartheta^* + \vartheta(\sigma_x\otimes \ph(v))\vartheta^* )  }  
	\end{align} 
	To prove that \cref{eq:Htilde} is positivity improving,
	 we use a perturbative argument which can be found in \cite[Theorem XIII.45]{ReedSimon.1978}. Explicitly,
	 % let $\chi_n=\chr{(-n,n)}$ be the characteristic function of the interval $(-n,n)$.
	%,
	by the definition of $\vartheta$ and \cref{prop:Q}, \[V_n := \vartheta(\sigma_x\otimes \ph(v)\chr{(-n,n)}(\ph(v)))\vartheta^*\] is a bounded multiplication operator in   $L^2(\{1,2\}\times \cQ_{L^2 (\IR^d;\IR)})$  for all $n \in \N$. Furthermore,  by the boundedness of $\ph(v)$ w.r.t. $\dG(\omega)$ (cf. \cref{lem:Fockprops}), we find
	$\vartheta \widetilde{H}(0,\mu) \vartheta^*+\lambda V_n$ converges to $\vartheta \widetilde{H}(\lambda,\mu) \vartheta^*$ in strong resolvent sense 
	and  $\vartheta \widetilde{H}(\lambda,\mu) \vartheta^* - \lambda V_n$ converges to $\vartheta \widetilde{H}(0,\mu)\vartheta^*$ in strong resolvent sense. 
	Hence, by \cite[Theorems XIII.43,XIII.45]{ReedSimon.1978} it follows that   \eqref{eq:Htilde}  is positivity improving  if and only if   
	\begin{align} \label{eq:pos} 
		e^{-T \vartheta \Ht(0,\mu)\vartheta^*} =  e^{ -T( ( 1- \sigma_x)+ \mu \sigma_z)  }     e^{-T  ( \Theta_{L^2 (\IR^d;\IR)}  \dG(\omega)\Theta_{L^2 (\IR^d;\IR)}^*  )  }
	\end{align}
	is. Note that  in  \eqref{eq:pos} the first factor only acts on the variables $\{1,2\}$, and the second factor only acts on the variables in $\cQ_{L^2(\IR^d;\IR)}$. 
	It is well-known (cf. \cite[Theorem I.16]{Simon.1974}) that the second factor on the right hand side of \eqref{eq:pos} 
	% $\Theta_{L^2 (\IR^d;\IR)}e^{-T\dG(\omega)}\Theta_{L^2 (\IR^d;\IR)}^*$
	is positivity improving on $L^2(\cQ_{L^2(\IR^d;\IR)})$. Further, by explicit computation, we have that the first factor on the right hand side of \cref{eq:pos} 
	\begin{equation*} \label{expT} 
		\exp\left( -T( ( 1- \sigma_x)+ \mu \sigma_z)  \right) = \begin{pmatrix}
			e^{-(\mu+1) T} & e^{T} \\ e^{(\mu+1) T} & e^{-T}
		\end{pmatrix}
	\end{equation*} 
	is positivity improving  on $L^2(\{1,2\})$, since all matrix elements are strictly positive.
	This finishes the proof.
\end{proof}
We now obtain \cref{lem:Bloch.1} as an easy corollary of \cref{lem:posimp}.
\begin{proof}[\textbf{Proof of \cref{lem:Bloch}}]
	Let $\vartheta$ be defined as in \cref{lem:posimp}.
	By  the definitions \cref{def:H,def:Ht}, we have
	\begin{align*}  \Braket{\Od,e^{-TH(\lambda,\mu)}\Od}_\HS &  = \Braket{(U\otimes\Id)^*\Od,e^{-T(\Ht(\lambda,\mu)- \one)}(U\otimes\Id)^*\Od}_\HS \\
		&  = \Braket{ \vartheta (U\otimes\Id)^*\Od, \big( \vartheta e^{-T(\Ht(\lambda,\mu)- \one)} \vartheta^* \big)  \vartheta (U\otimes\Id)^*\Od}_{   L^2(\{1,2\}\times \cQ_{L^2 (\IR^d;\IR)}) }   .  
	\end{align*} 
	By \eqref{def:unitary}, \eqref{def:vac},  and \eqref{prop:Q}, we see that   $\vartheta(U\otimes\Id)^*\Od = 2^{-1/2}$, which  is  a constant  strictly positive  function.
	Hence, %if  we show that 
	%\cref{eq:Htilde}  (and thus $\vartheta e^{-T(\Ht(\lambda,\mu) - \one)} \vartheta^*$) is positivity improving,
	the statement follows from \cref{lem:posimp,applem:Bloch2}, since  $\inf \sigma(H(\lambda,\mu)) = \inf \sigma (\widetilde{H}(\lambda,\mu)- \one)$.
\end{proof}
Further, the following statement also is a direct consequence of \cref{lem:posimp}. It will be a useful ingredient to our proof of \cref{prop:partition}.
\begin{prop}\label{lem:Bloch.2}
	If $E(\lambda,\mu)$ is an eigenvalue of $H(\lambda,\mu)$, then the corresponding eigenspace is non-degenerate. In this case, if $\psi_{\lambda,\mu}$ is a ground state of $H(\lambda,\mu)$, then $\braket{\psi_{\lambda,\mu},\Od}\ne 0$.
\end{prop}
\begin{proof}%[\textbf{Proof of \cref{lem:Bloch,lem:Bloch.2}}]
%	The proof of both statements uses positivity arguments. To that end, we first need to unitarily map $\HS=\IC^2\otimes \FS$  to an $L^2$-space. 
%	 For this 
%	let $\vartheta$ be the natural Hilbert space isomorphism $\HS=\IC^2\otimes \FS\to L^2(\{1,2\}\times \cQ_{L^2 (\IR^d;\IR)})$ which is determined by
%	\[ \alpha \otimes \psi \mapsto \left((i,x)\mapsto \alpha_i (\Theta_{L^2 (\IR^d;\IR)} \psi) (x) \right),  \]
%	where $\Theta_{L^2 (\IR^d;\IR)}:\FS\to L^2(\cQ_{L^2 (\IR^d;\IR)})$ is the unitary from
%	\cref{prop:Q}.
%	
%	It suffices to prove that
%	\begin{align}  \label{eq:Htilde} \vartheta e^{-T\Ht(\lambda,\mu)} \vartheta^*  =  e^{ - T \vartheta \tilde{H}(\lambda,\mu)\vartheta^*} 
%		=  e^{ - T (  \vartheta \widetilde{H}(0,\mu)\vartheta^* + \vartheta(\sigma_x\otimes \ph(v))\vartheta^* )  }  
%	\end{align}
%	 is positivity improving for all $T>0$.
	
	 By the Perron-Frobenius-Faris theorem \cite[Theorem XIII.44]{ReedSimon.1978} and \cref{lem:posimp}, if $E(\lambda,\mu)$ is an eigenvalue of $H(\lambda,\mu)$, then there exists a strictly positive $\phi_{\lambda,\mu}\in L^2(\{1,2\}\times \cQ_{L^2 (\IR^d;\IR)})$ such that the eigenspace corresponding to $E(\lambda,\mu)$ is spanned by $\vartheta(U\otimes \Id)^*\phi_{\lambda,\mu}$, where $\vartheta$ again is the defined as in \cref{lem:posimp}.
	 Since $\vartheta (U\otimes \Id)^* \Od$ is (strictly) positive, this proves the statement.
\end{proof}
We can now prove the Bloch formula for the derivatives.
\begin{proof}[\textbf{Proof of \cref{prop:partition}}]
	Throughout this proof, we fix $\lambda,\mu_0$ as in the statement of the \lcnamecref{prop:partition}.
	\newcommand{\bh}{\mathbf{h}}\newcommand{\be}{\mathbf{e}}
	Further, for compact notation, we write
	\[\bh(\mu)=H(\lambda,\mu), \quad \be(\mu)=E(\lambda,\mu) \quad\mbox{and}\quad \be_T(\mu) = -\frac 1T\ln\Braket{\Od,e^{-Th(\mu)}\Od}.\]
	Hence, we want to prove
	\[ \be^{(n)}(\mu_0) = \lim_{ T \to \infty}\be_T^{(n)}(\mu_0) \qquad\mbox{for all}\ n\in\IN,\]
	where $(\cdot)^{(n)}$ as usually denotes the $n$-th derivative.
	
	We observe that the ground state energy  $\be(\mu_0)$ is a simple eigenvalue of $\bh(\mu_0)$, by \cref{lem:Bloch.2}. Further,
	by view of  \cref{def:H}, it is obvious that  the operator valued family $\mu\mapsto \bh(\mu) = \bh(0)+\mu\sigma_x\otimes \Id $ is  an analytic family of type (A), cf. \cite{Kato.1980,ReedSimon.1978}.
	Then, by the Kato-Rellich theorem \cite[Theorem XII.8]{ReedSimon.1978}, it follows that $\mu\mapsto \be(\mu)$ is analytic and $\be(\mu)$ is an isolated simple eigenvalue  of $\bh(\mu)$ in a neighborhood of $\mu_0$.
	
	We introduce  the distance of $\be(\mu_0)$ to the rest of the spectrum   by  
	\[\delta = \operatorname{dist}(\be(\mu_0),\sigma(\bh(\mu_0))\setminus\{\be(\mu_0)\}).\]
	By the Kato-Rellich theorem \cite[Theorem XII.8]{ReedSimon.1978},  we  can choose  an $\eps>0$ such that
	\begin{equation}\label{eq:epsassumptions}
		\left|\be(\mu)-\be(\mu_0)\right|\le \frac \delta 4 \qquad\mbox{and}\qquad \inf(\sigma(\bh(\mu))\setminus \be(\mu)) \ge \be(\mu_0)+ \frac 34 \delta \qquad\mbox{for}\ \mu\in(\mu_0-\eps,\mu_0+\eps), 
	\end{equation}
where the second inequality  can be obtained  using a Neumann series,  cf. \eqref{eq:series}, or alternatively it can be obtained from  the   lower boundedness of  \cref{lem:sbmin} and a compactness argument involving that the set of $(\mu,z)$, for which $\bh(\mu)-z$ is invertible,  is open, see   \cite[Theorem XII.7]{ReedSimon.1978}.
%	\[  B_{\frac{\delta}{2}}(e(\mu_0))\cap \sigma(h(\mu))=\{e(\mu)\}, \]
%	where $B_{\frac{\delta}{2}}(\mu)$ denotes the open ball of radius $\frac \delta 2$ centered in $e(\mu_0)$ in the complex plane.
%(Alternatively, one  could use a Neumann expansion of the resolvent,  \eqref{eq:series},  to show \eqref{eq:epsassumptions}.) 
	Henceforth, let  $\mu\in (\mu_0-\eps,\mu_0+\eps)$. Then, by  \cref{eq:epsassumptions}, we can write the ground state projection $P(\mu)$ of $\bh(\mu)$
 as
%	
%	For $\mu$ in a neighborhood of $\mu_0$, we denote by $P(\mu)$  the ground state projection of $h(\mu)$, i.e.,
	\[
	P(\mu) = \frac{1}{2\pi \i} \int_{\Gamma_0}  \frac{1}{z- \bh(\mu)} dz ,
	\]
	where $\Gamma_0$ is a curve   encircling  counterclockwise the point  $\be(\mu_0)$  at a distance $\delta/2$ .
	Further, let
	\begin{align*}
		& \gamma_0 : [-1,+1] \to \C  , \quad t \mapsto \be(\mu_0) + \frac{\delta}{2} - \i t, \\
		& \gamma_\pm : [0,\infty) \to \C , \quad t \mapsto \be(\mu_0) + \frac{\delta}{2} \pm  \i  + t 
	\end{align*}
	and define the curve $\Gamma_1 = -  \gamma_+ +  \gamma_0 +  \gamma_- $ surrounding  the set $\sigma(\bh(\mu_0)) \setminus \{ \be(\mu_0) \}$ (see \cref{fig1}).
	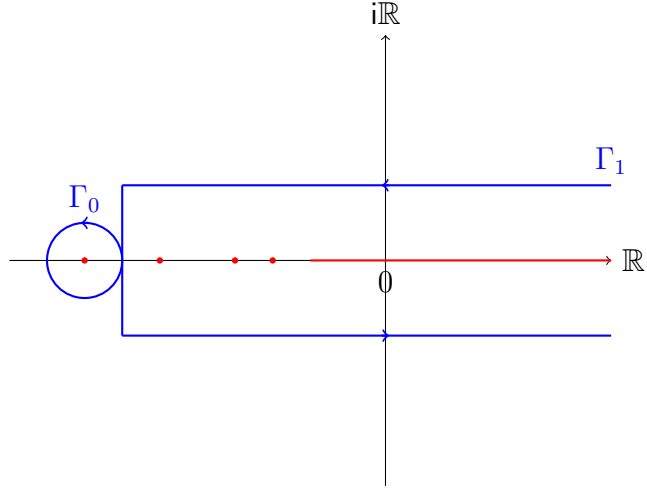
\begin{figure}
	\begin{center} 
		\begin{tikzpicture}
			\draw[thick, blue]  (-4,0)  circle[radius=0.5] ;
			\draw[blue,<-,thick] (-4.05,0.5)  -- (-3.95,0.5)  ;
			\draw[blue]  (-4,0.5)  node[above] {$\Gamma_0$};
			\draw[->] (-5, 0) -- (3, 0) node[right] {$\IR$};
			\draw[->] (0,  - 3) -- (0,3) node[above] {$\i \IR$};
			% \draw[ domain=-1.8:1.8, smooth, variable=\x, blue, thick] plot ({\x}, {\x*\x});
			\draw[blue, thick] (-3.5,-1) -- (-3.5,1) ;  
			\draw[blue, thick] (-3.5,1) -- (3,1) node[above] {$\Gamma_1$};  
			\draw[blue, thick] (-3.5,-1) -- (3,-1) ;  
			\draw[blue,<-,thick] (-0.05,1)  -- (0.05,1)  ;
			\draw[blue,->,thick] (-0.05,-1)  -- (0.05,-1)  ;
			%  \draw[ blue, thick] (-4, 0) -- (-5,0) ;  
			%\draw[->, blue, thick] (1.2, 0.1) -- (1.8,0.1) ;\draw[->, blue, thick] (1.2, 0.1) -- (1.8,0.1) ;
			%  \draw[->, blue, thick] (2.2, 0.1) -- (2.8,0.1) ;
			%  \draw[<-, blue, thick] (-0.2, 0.1) -- (-0.8,0.1) ;  \draw[<-, blue, thick] (-1.2, 0.1) -- (-1.8,0.1) ;
			%  \draw[<-, blue, thick] (-2.2, 0.1) -- (-2.8,0.1) ;
			\filldraw[red]  (-4,0)  circle[radius=1pt];
			\filldraw[red]  (-3,0)  circle[radius=1pt];
			\filldraw[red]  (-2,0)  circle[radius=1pt];
			\filldraw[red]  (-1.5,0)  circle[radius=1pt];
			\draw[red,thick]  (-1,0) -- (3,0) ;
			\draw  (0,0)   node[below] {$0$};
			% \draw (.5,.5) pic[red,thick] {cross=0.1};
			%   \draw[red, ->] (0.5,0) arc (0:-120:0.5) node[below , right] {$2 {\rm Im} \theta$};
			%    \draw[blue,->] (-4.5,0) arc (180:120:0.5) node[left] {$ {\rm Im} \theta$}  ;
		\end{tikzpicture}
	\end{center}
	\caption{Illustration of the curves $\Gamma_0$ and $\Gamma_1$. The spectrum of $\bh(\mu_0)$ is displayed in red.}
	\label{fig1}
\end{figure}
	In view of \cref{eq:epsassumptions}, we can define 
	\[Q_T(\mu) :=  \frac{1}{2\pi \i} \int_{\Gamma_1}  \frac{e^{ T (  \be(\mu) - z )}}{z- \bh(\mu)} dz  , \]
where the integral is understood as a Riemann integral with respect to the   operator topology. 	
The spectral theorem for the self-adjoint operator $\bh(\mu)$  and Cauchy's integral formula  yield
	\begin{equation}\label{eq:projectionsum}
		e^{-T(\bh(\mu)-\be(\mu))} = P(\mu) + Q_T(\mu).
	\end{equation}
	For  $z \in \rho(\bh(\mu_0))$  and  $\mu$ in a neighborhood of $\mu_0$  we have 
	\begin{equation}\label{eq:series}
		\frac{1}{z-\bh(\mu)} = \frac{1}{z-\bh(\mu_0)}\sum_{k=0}^{\infty}\left((\mu-\mu_0)(\sigma_x\otimes\Id)\frac{1}{z-\bh(\mu_0)}\right)^k . 
	\end{equation}
	Using this expansion and the following bounds obtained from \cref{eq:epsassumptions}
 \begin{equation}\label{eq:bounds}
		\begin{aligned}
			&\|(z-\bh(\mu_0))^{-1}\|\le  \frac 2  \delta   \quad   \mbox{for}\ z\in  \ran \Gamma_0  \cup \ran \Gamma_1 ,  \\
	&| e^{T(\be(\mu)-z)}| \le \begin{cases} e^{-\frac{\delta}{4} T}  & \mbox{for}\ z\in \ran \gamma_0, \\ e^{-\frac{\delta}{4} T}  e^{-Tt} & \mbox{for}\ z=\gamma_\pm(t), \  t \in [0, \infty) ,  \end{cases}
		\end{aligned} 
	\end{equation} 	
%	\begin{equation}\label{eq:bounds}
%		\begin{aligned}
%			&\|(z-h(\mu))^{-1}\|\le \begin{cases} \frac \delta 4 & \mbox{for}\ z\in \gamma_0, \\ 1 & \mbox{for}\ z\in\gamma_\pm,\end{cases}\\
%			&| e^{T(e(\mu)-z)}| \le \begin{cases} e^{-\frac{\delta}{4} T}  & \mbox{for}\ z\in\gamma_0, \\ e^{-\frac{\delta}{4} T}  e^{-Tt} & \mbox{for}\ z=\gamma_\pm(t)., \end{cases}
%		\end{aligned}
%	\end{equation} 
	we see that $P(\mu)$ and $Q_T(\mu)$ are  real analytic  for $\mu$ in a neighborhood of $\mu_0$ and, moreover,  that the integrals and derivatives with
respect to $\mu$   can be interchanged due to the uniform convergence of the integrand on the curves $\Gamma_0$ and $\Gamma_1$. Hence,
	%
	%We now restrict ourselves to real values of $\mu$. Then, 
	by virtue of \cref{eq:projectionsum}, we see
	   that the function  $\mu\mapsto\Braket{\Od,e^{-\bh(\mu)}\Od}$ is  real analytic  on $(\mu_0-\wt\eps,\mu_0+\wt\eps)$ for $\wt\eps\in(0,\eps)$ small enough.
	   
%	Using the positivity argument from above proof, we can choose the ground state $\psi_\mu$ of $h(\mu)$ such that $\vartheta\psi_\mu\in L^2(\{1,2\}\times \cQ_{L^2 (\IR^d;\IR)})$ is strictly positive, cf. \cite[Theorem XIII.44]{ReedSimon.1978}. Hence, $\braket{\psi_\mu,\Od}=\braket{\vartheta\psi_\mu,\vartheta\Od} >0$ and w
	Let $\psi_\mu$ be a normalized  ground state of $\bh(\mu)$. Then, by \cref{lem:Bloch.2}, we find
	\begin{equation}\label{eq:positivegs}
		\Braket{\Od,P(\mu)\Od} = |\braket{\psi_\mu|\Od}|^2>0.
	\end{equation}
	%\textcolor{red}{\texttt{Move to Lemma above!}}
	Further, by the spectral theorem and \cref{eq:epsassumptions}
	\begin{equation}\label{eq:boundQT}
		0\le \Braket{\Od,Q_T(\mu)\Od}=e^{-\frac 12T\delta}\int_{\be(\mu)+\frac12\delta}^{\infty}e^{T(\be(\mu)+\frac 12\delta-\lambda)}d\nu_{\Od}(\lambda)\le e^{-\frac 12T\delta}\|\Od\|^2,
	\end{equation}
	where $\nu_{\Od}$ denotes the spectral measure of $h(\mu)$ associated with $\Od$, cf. \cite[Section VII.2]{ReedSimon.1972}.
	
	By \cref{eq:projectionsum} and the definition of $\be_T(\mu)$, we have
	\[ \be(\mu) -  \be_T(\mu)  =   \frac 1T\ln \left(\Braket{\Od,P(\mu)\Od} + \Braket{\Od,Q_T(\mu)\Od}\right) \qquad\mbox{for}\ \mu\in(\mu_0-\wt{\eps},\mu_0+\wt{\eps}).\]
	Hence, we can calculate the $n$-th derivative of the expression on the left hand side at $\mu=\mu_0$, by taking the $n$-th derivative on the right hand side.
	Using the Faà di Bruno formula (\cref{lem:faa}) and recalling the notation from \cref{david123}, we  find
	\begin{align*} 
		\be^{(n)}(\mu_0)- \be_T^{(n)}(\mu_0) = & \frac{-1}{T} \sum_{\fP \in \mathcal{P}_n} \frac{   (-1)^{|\fP|} (|\fP|-1)! }{ \langle \Od , ( P(\mu_0) + Q_T(\mu_0)) \Od \rangle^{|\fP|}} \prod_{B \in  \fP } 
		\Braket{\Od , ( P^{(\lvert B\lvert)}(\mu_0) + Q_T^{(\lvert B\lvert)}(\mu_0)) \Od} .
 	\end{align*}
	By \cref{eq:positivegs,eq:boundQT}, the first factor is uniformly bounded in $T$. Hence, it remains to prove that $\Braket{ \Od , Q_T^{(k)}(\mu_0) \Od}$ is uniformly bounded in $T$ for all $k=1,\ldots,n$.
	Therefore, we explicitly calculate the derivative of $Q_T(\mu)$ at $\mu=\mu_0$. This is done by interchanging the integral with the derivative, which we justified above.
	Note that, by the series expansion \cref{eq:series}, we have
	\[\partial_\mu^k(z-\bh(\mu))^{-1}= \frac{k!}{z-\bh(\mu)}\left(\sigma_x \frac{1}{z-\bh(\mu)}\right)^k \qquad\mbox{for}\ k\in\IN_0.\]
	Again using Faà di Bruno's formula (\cref{lem:faa}) and the Leibniz rule, this yields
	\begin{align*} 
		&Q_T^{(k)}(\mu_0)   =  \frac{1}{2\pi \i}   \sum_{\ell=0}^k \binom{k}{\ell} \int_{\Gamma_1}  \partial_\mu^\ell ( e^{ T (  \be( \mu) - z )} ) \partial_\mu^{k-\ell} (z- \bh( \mu))^{-1} dz {\bigg|_{\mu = \mu_0}} \nonumber \\
		&\quad=  \frac{1}{2\pi \i}   \sum_{\ell=0}^k \binom{k}{\ell} (k-\ell)! \underbrace{\left(\sum_{\fP \in \mathcal{P}_\ell} \prod_{B \in \fP} (   T \be^{(|B|)}(\mu_0) )\right)}_{=: P_{k,\ell}(T)}  %\label{eqes:pre}
		%\\ &\qquad \times
		\underbrace{\int_{\Gamma_1} e^{ T (  \be(\mu_0) - z )}  { \frac{1}{z- \bh(\mu_0)} }   \left( \sigma_x \frac{1}{z- \bh(\mu_0)} \right)^{k-\ell} dz}_{=: I_{k,\ell}(T)}.% \label{eqes:int}
	\end{align*}
	%\oli{$e^{(|B|)}$ ist unguenstige Notation, keine Unterscheidung zu Exponential. Schreibe vielleicht lieber $\partial_\mu|_{\mu = \mu_0} e(\mu)$}
	Applying the bounds \cref{eq:bounds}, we find
	\[\|I_{k,\ell}(T)\| \le \left(\frac2\delta\right)^{k-\ell}e^{-\frac\delta4T} \left[ \int_{-1}^11dt + 2\int_0^\infty e^{-Tt}dt \right] .\]
	Since $P_{k,\ell}(T)$ only grows polynomially in $T$, this implies $\|Q_T^{(k)}(\mu_0)\|\xrightarrow{T\to\infty}0$ and especially proves $\Braket{ \Od , Q_T^{(k)}(\mu_0) \Od}$ is uniformly bounded in $T$.
\end{proof}
We now combine Bloch's formula for derivatives of the ground state energy with the FKN formula.
\begin{proof}[\textbf{Proof of \cref{david123}}]
	First, we recall the definition of $Z_T(\lambda,\mu)$ in \cref{def:ZT} and the notation $\bn{\cdot}_{T,\lambda,\mu}$ from \cref{def:exp}.
	By the dominated convergence theorem, one sees that $Z_T$ is infinitely often differentiable in $\mu$
%	Calculating explicitly for paths $x\in\Df$ and using the dominated convergence theorem, it follows that $Z_T$ is analytic in $\mu$
	and has the derivatives
	\begin{equation}\label{eq:derivativeZT}
			\begin{aligned}
			\partial_\mu^nZ_T(\lambda,\mu)
			& = (-1)^n\EE_X\left[\left( \int_0^T X_t dt\right)^n\exp\left(\lambda^2\int_0^T\int_0^T W(t-s)X_tX_sdsdt {-} \mu\int_0^TX_tdt\right)\right]
			\\& = (-1)^nZ_T(\lambda,\mu)\left\llangle \left( \int_0^T X_t dt\right)^n \right\rrangle_{T,\lambda,\mu} .
		\end{aligned}
	\end{equation}
%	Further,
%	combining the results from \cref{cor:intout,lem:Bloch}, we see that \cref{eq:intout0} holds. 
	Further, first using \cref{prop:partition} and the Fa\`a di Bruno formula \cref{lem:faa} to calculate the derivatives of the logarithm yields
	\begin{equation}\label{eq:derivativeGSen}
		\begin{aligned}
			\partial_\mu^n E(\lambda,\mu) 
			& = - \lim_{T \to \infty} \frac1T \sum_{\fP\in\cP_n} \frac{(-1)^{|\fP|-1}(|\fP|-1)!}{(\Braket{\Od,e^{-TH(\lambda,\mu)}\Od})^{|\fP|}} \prod_{B \in  \fP } \partial_\mu^{|B|} \Braket{\Od,e^{-TH(\lambda,\mu)}\Od}\\
			&= - \lim_{T \to \infty} \frac1T \sum_{\fP\in\cP_n} \frac{(-1)^{|\fP|-1}(|\fP|-1)!}{(Z_T(\lambda,\mu))^{|\fP|}} \prod_{B \in  \fP } \partial_\mu^{|B|}Z_T(\lambda,\mu),
		\end{aligned}
	\end{equation}
	where in the last line we inserted the identity \cref{eq:intout0} (which in turn follows from \cref{cor:intout}).
	Combining \cref{eq:derivativeZT,eq:derivativeGSen} proves \cref{david123}.
	\end{proof} 
	
	\begin{proof}[\textbf{Proof of \cref{david1234}}] 
	By  the definition of the Ursell functions \cref{def:ursell}  we find using  the Fa\`a die Bruno formula  
	\[ u_n\left(\int_0^T X_{s_1}ds_1,\ldots, \int_0^T X_{s_n}ds_n\right)  = 
	 \sum_{\fP  \in \mathcal{P}_n} (-1)^{|\fP|+n} ( |\fP|  -1)!  \prod_{B \in  \fP  } 
		\left\llangle   \left( \int_0^T X_t dt \right)^{|B|}    \right\rrangle_{T,\lambda,\mu}  .
		\]
	Now, the Ursell functions are multilinear, cf. \cite[Section 11]{Percus.1975}, and by the dominated convergence theorem we can hence exchange the integrals with the expectation value, i.e., 
		\[ u_n\left(\int_0^T X_{s_1}ds_1,\ldots, \int_0^T X_{s_n}ds_n\right)  = \int_0^T \cdots \int_0^T  u_n\left( X_{s_1},\ldots, X_{s_n}\right)   ds_1 \cdots ds_n . 
		\]
	Inserting this into \cref{david123}  finishes the proof of \cref{david1234}.
\end{proof}

\section{Existence of Ground States}\label{sec:suszbound}

In this \lcnamecref{sec:suszbound}, we use  the  bound on the second derivative of the ground state energy as function of the magnetic coupling from \cref{cor:corr}
to obtain the result that the spin boson Hamiltonian with massless bosons has a ground state for couplings  which exhibit  strong infrared singularities. 
This result  is non-trivial, since  the massless bosons imply that  there is  no spectral gap.

Our main result needs the following assumptions.
\begin{hyp}\label{hyp:gsexists}\ 
 	\begin{enumhyp}
		\item $\nu :\IR^d\to[0,\infty)$ is locally H\"older  continuous, positive a.e.,  and $\nu(k)=\nu(-k)$.
		\item $\lim\limits_{|k|\to\infty}\nu(k)=\infty$.
		\item\label{part:eps}  $v\in L^2(\IR^d)$ has real Fourier transform and there exists $\eps>0$, such that $\nu^{-1/2}v\in L^{2+\eps}(\IR^d) \cap L^2(\R^d)$.
		\item\label{part:integrals} $\displaystyle\sup_{|p|\le 1}\int_{\IR^d}\frac{|v(k)|}{\sqrt{\nu(k)}\nu(k+p)}dk<\infty$ and $\displaystyle\sup_{|p|\le 1}\int_{\IR^d}\frac{|v(k+p)-v(k)|}{\sqrt{\nu(k)}|p|^\alpha}dk<\infty$  for some $\alpha > 0$ .
	\end{enumhyp}
\end{hyp}
We can now state the main result of this section. 
\begin{thm}\label{thm:gsexists}
	Assume \cref{hyp:gsexists} holds. 
	Then there exists $\lambda_\sfc>0$, such that for all $\lambda\in(-\lambda_\sfc,\lambda_\sfc)$
	the spin boson Hamiltonian 
	\begin{equation}\label{def:Hl}
		H_\lambda = \sigma_z\otimes\Id + \Id\otimes \dG(\nu) +\lambda  \sigma_x\otimes \ph(v) 
	\end{equation}
	acting on $\C^2 \otimes \FS$  has a ground state, i.e., the infimum of the spectrum is an eigenvalue.
\end{thm}
\begin{rem}
	This improves the previously known results on ground state existence \cite{HaslerHerbst.2010,BachBallesterosKoenenbergMenrath.2017}.
	We also remark that in principal our method of proof not only gives existence of a small $\lambda_\sfc$, but could in fact be used to estimate the critical coupling constant, due to its non-perturbative nature.
\end{rem}
\begin{ex}
	Let us consider the case
	\begin{equation}\label{def:physics}
		d=3, \qquad \nu(k)= |k| \qquad\mbox{and}\qquad v(k)=\chi(k)|k|^\delta,
	\end{equation}
	where $\chi:\IR^d\to\IR$ is the characteristic function of an arbitrary ball around $k=0$. Obviously the assumptions on $\omega$ in \cref{hyp:gsexists} are satisfied. Further, \cref{part:eps} holds for any $\delta>-1$ as is easily verified by integration in polar coordinates. The finiteness conditions \subcref{part:integrals} of \cref{hyp:gsexists} also hold in this case by simple estimates. 
	We remark that the previous results \cite{HaslerHerbst.2010,BachBallesterosKoenenbergMenrath.2017} covered the situation \cref{def:physics} with $\delta=-\frac 12$.
\end{ex}

The method of proof relies on the approximation of the photon dispersion relation $\nu$ by the infrared-regularized versions $\nu_m= \sqrt{\nu^2+m^2}$ with $m>0$. We denote by $H_m$ and $E_m$ the definitions \cref{def:H,def:E} with $\omega$ replaced by $\nu_m$. Since $\inf_{k\in\IR^d}\nu_m(k)\ge m >0$, the operator $H_m(\lambda,0)$ has a spectral gap for any $m>0$ and hence also a ground state, cf. \cref{thm:massive}.
In the recent paper \cite{HaslerHinrichsSiebert.2021a}, we showed the following result, which together with  \cref{cor:corr} give a proof of \cref{thm:gsexists}. 
\begin{thm}[{\cite{HaslerHinrichsSiebert.2021a}}]\label{prop:gsexistsold}
	Assume \cref{hyp:gsexists} holds and let $\lambda\in\IR$. If 
	for small  $m > 0$ the function 
	$\mu \mapsto E_m(\lambda,\mu)$ is twice differentiable  at zero  and $\limsup_{m\downarrow 0} | \partial_\mu^2 E_m(\lambda,0)|<\infty$, then $H_\lambda$ as defined in \cref{def:Hl} has a ground state.
\end{thm}\noindent
We conclude with the proof for existence of ground states.
\begin{proof}[\textbf{Proof of \cref{thm:gsexists}}] Applying  \cref{cor:corr} to the function $\nu$, we see that the assumptions of  \cref{prop:gsexistsold}
	are satisfied. This proves the theorem.
\end{proof}

%\subsection*{Acknowledgements}

\appendix

\section[L2-valued Riemann Integral and  Pointwise Lebesgue Integrability]{$L^2$-valued Riemann Integral and  Pointwise Lebesgue Integrability}

\label{app:integrals}
In \cref{rem:integrals}, we use the following \lcnamecref{lem:integrals} with $f(t)=\pE(j_tv)x_t$.
\begin{lem}\label{lem:integrals}
	Let $(\cQ,\mu)$ be a probability space and assume $t\mapsto f_t\in L^2(\cQ)$ is piecewise continuous on the interval $[0,T]$. Then $[t\mapsto f_t(q) ] \in L^1([0,T])$ for almost every $q\in\cQ$ and
	\begin{equation}\label{eq:intequal}
		\int_0^T f_t(q)dt = \left(\int_0^T f_tdt\right)(q) \qquad\mbox{for almost every}\ q\in\cQ, 
	\end{equation}
	where the integral on the right hand side is the $L^2(\cQ)$-valued Riemann integral.
\end{lem}
\begin{proof}
	Using Fubini's theorem and Hölder's inequality, we find
	\[ \int_{\cQ} \int_0^T |f_t(q)|dt d\mu(q) = \int_0^T \int_{\cQ} |f_t(q)|d\mu(q) dt 
	%
	 %\left\| \int_0^T f_tdt\right\|_{L^1(\cQ)} \le \int_0^T \|f(t)\|_{L^1(\cQ)} 
	 \le \int_0^T \|f_t\|_{L^2(\cQ)}dt<\infty.  \]
	Hence, for $\mu$-almost all $q\in\cQ$, the map $t\mapsto f_t(q)$ is Lebesgue-integrable. Let $f_{n,t}$ be an $L^2(\cQ)$--valued step function. 
Then using the triangle inequality, Fubini's theorem and  H\"older's inequality, we find 
		\begin{align*} 
		&  \int_{\cQ} \left|\int_0^Tf_t(q)dt- \left( \int_0^T f_{t}dt\right)(q) \right| d\mu(q) \\
		&\le   \int_{\cQ} \left|\int_0^Tf_t(q)dt-  \int_0^T f_{n, t}(q)dt  \right| d\mu(q)  +    \int_{\cQ} \left| \left(  \int_0^T f_{n, t}\right)(q)dt - \left( \int_0^T f_{t}dt\right)(q) \right| d\mu(q)  \\ 
		& \le  \int_{\cQ} \int_0^T\left|f_t(q)- f_{n,t}(q)\right|dt d\mu(q) +   \left\|\int_0^T f_tdt-\int_0^T f_{n,t}dt\right\|_{L^1(\cQ)}  \\
		&\le  \int_0^T  \int_{\cQ}\left|f_t(q)- f_{n,t}(q)\right| d\mu(q)dt  +   \left\|\int_0^T f_tdt-\int_0^T f_{n,t}dt\right\|_{L^2(\cQ)}   \\
		& \le 2 \int_0^T \|f_t-f_{n,t}\|_{L^2(\cQ)}dt  . %+  \int_0^T \|f_t-f_{n,t}\|_{L^2(\cQ)}dt  .
		%&       \int d\mu_E(q)  \left|  \int_0^T ( f_t(x)  dt - \int_0^T f_{n,t} dt \right| \\ & \leq    \int d\mu_E(q) 	        \int_0^T |   ( \pE\left( j_t v\right))(q)x_t  - f(t,q) |   dt  \\
%		& \leq  	      \int_0^T   \int  d\mu_E(q)   |   ( \pE\left( j_t v\right))(q)x_t  - f(t,q) |   dt  \\      
%		& \leq  	      \int_0^T   \left(  \EEE  |    \pE\left( j_t v\right)x_t  - f(t) |^2  \right)^{1/2}   dt    \to 0 , 
	\end{align*} 
	Now, by the piecewise $L^2(\cQ)$-continuity of $t \mapsto f_t$, the right hand side can be made arbitraritly small by making   the mesh of  the Riemann
 sum arbitrarily small. 
	This implies \cref{eq:intequal}. 
\end{proof}

\section{The Fa\`a di Bruno Formula}\label{app:faa}

The following formula is used in several places throughout the paper. A proof and historical discussion can be found in \cite{Hardy.2006}.
\begin{lem}\label{lem:faa} 
	Let $I \subset \IR$ and $\Omega \subset \R^m$ be open and let $f:J\to \IR$ and $g: \Omega  \to J$ be $n$-times continuously differentiable functions. Then $f\circ g: \Omega  \to\IR$ is $n$-times continuously differentiable and for any choice of  $k_1,\ldots,k_n \in \{1,\ldots,m\}$%, we have 
	\[ \frac{ \partial^n}{ \partial x_{k_1} \cdots \partial x_{k_n}  }  (f\circ g)  = \sum_{\fP\in \cP_n}(f^{(|\fP|)}\circ g) \prod_{B\in \fP}   \frac{ \partial^{|B|} g  }{ \prod_{j \in B} \partial x_{k_j} } , \]
	where $\cP_n$ denotes the set of partitions of the set $\{1,\ldots,n\}$.
\end{lem}

\section[Fock Space and Q-Space]{Fock Space and $\cQ$-Space}\label{app:FQ}

\subsection{Standard Fock Space Properties}\label{app:Fockspace}

In this \lcnamecref{app:Fockspace}, we collect well-known properties of the Fock space operators introduced in \cref{subsec:SBmodel}. In large parts, these can be found in standard textbooks such as \cite{ReedSimon.1975,Parthasarathy.1992,BratteliRobinson.1996,Arai.2018}. For the convenience of the reader, we give exemplary precise references to \cite{Arai.2018} below.
\begin{lem}\label{lem:Fockprops}
	Let $\hs,\fv,\fw$ be complex Hilbert spaces, let $A$ be a self-adjoint operator on $\hs$, let $B:\fv\to\fw$ and $C:\hs\to\fv$ be contraction operators and let $f,g\in \hs$.
	\begin{enumlem}
		\item\label{lem:Fockprops.SA} $\dG(A)$ is self-adjoint and $e^{\i t\dG(A)}=\Gamma(e^{\i tA})$.
		\item\label{lem:Fockprops.pos} If $A\ge 0$, then $\dG(A)\ge 0$ and $e^{-t\dG(A)}=\Gamma(e^{-tA})$ for $t\ge 0$.
		\item\label{lem:Fockprops.cont} $\Gamma(B)$ is a contraction operator.
		\item\label{lem:Fockprops.unit} If $B$ is unitary, so is $\Gamma(B)$.
		\item\label{lem:Fockprops.prod} $\Gamma(B)\Gamma(C)=\Gamma(BC)$ and $\Gamma(B)^* = \Gamma(B^*)$.
		\item\label{lem:Fockprops.SA2} $\ph(f)$ is self-adjoint.
		\item\label{lem:Fockprops.standardbound} If $A\ge0$ is injective and $f\in\cD(A^{-1/2})$, then $\ph(f)$ is $\dG(A)^{1/2}$-bounded and for all $\psi\in\cD(\dG(A)^{1/2})$
		\[\|\ph(f)\psi\| \le \sqrt 2\|A^{-1/2}f\|\|\dG(A)^{1/2}\psi\| + \frac{1}{\sqrt 2}\|f\|\|\psi\|.\]
			Especially, in this case $\ph(f)$ is infinitesimally $\dG(A)$-bounded.
		\item\label{applempart:trans1} $\Gamma(B)a^*(f)=a^*(Bf)\Gamma(B)$ and $a(f) \Gamma(B)^*= \Gamma(B)^*a(Bf) $.
		\item\label{appleampart:trans2} If $B$ is an isometry, i.e., $B^*B = \Id_{\hs}$, then $\Gamma(B)a(f)=a(Bf)\Gamma(B)$ and hence $\Gamma(B)\ph(f) = \ph(Bf)\Gamma(B)$.
	\end{enumlem}
\end{lem}
{\em References in \cite{Arai.2018}.}\\
	For \subcref*{lem:Fockprops.SA}, \subcref*{lem:Fockprops.pos} see Theorems 5.2 and 5.7. For \subcref*{lem:Fockprops.cont} see Theorem 5.5. For \subcref*{lem:Fockprops.unit}, \subcref*{lem:Fockprops.prod} see Theorem 5.6. For \subcref*{lem:Fockprops.SA2} see Theorem 5.22. For \subcref*{lem:Fockprops.standardbound} see Proposition 5.12. The last sentence in \subcref*{lem:Fockprops.standardbound} follows from the inequality $\|\dG(A)^{1/2}\psi\|\le \eps \|\dG(A)\psi\| + \frac 1\eps \|\psi\|$, which holds for any $\eps>0$.
% \begin{enumproof}
%	\item[For {\subcref*{lem:Fockprops.SA},\subcref*{lem:Fockprops.pos}}]\addtocounter{enumproofi}{2}\! see Theorems 5.2 and 5.7.
%	\item[For (iii)] \! see Theorem 5.5.
%	\item[For (iv),(v)]\addtocounter{enumproofi}{2}\! see Theorem 5.6.
%	%\item[{\subcref*{lem:Fockprops.cont},\subcref*{lem:Fockprops.unit}}]\addtocounter{enumproofi}{2}\! Theorem 5.6
%	\item[For (vi)]\! see Theorem 5.22.
%	\item[For (vii)]\! see Proposition 5.12. The last sentence in (vii) follows from the inequality $\|\dG(A)^{1/2}\psi\|\le \eps \|\dG(A)\psi\| + \frac 1\eps \|\psi\|$, which holds for any $\eps>0$.
%\end{enumproof}
\begin{proof}[Proof of \subcref{applempart:trans1} and \subcref{appleampart:trans2}]
	%It remains to prove \subcref{applempart:trans1} and \subcref{appleampart:trans2}.
	\subcref{applempart:trans1} follows by observing
	\begin{align*}
		\Gamma(B)a^\dag(f)g_1\os\cdots \os g_n
		&= \sqrt{n+1} \Gamma(B) f\os g_1\cdots\os g_n = \sqrt{n+1}Bf\os Bg_1\cdots\os Bg_n\\
		& = a^\dag(Bf)Bg_1\os\cdots \os Bg_n = a^\dag (Bf)\Gamma(B)g_1\os\cdots\os g_n
	\end{align*}
	and using that the span of vectors of the form $g_1\os \cdots g_n$ is a core for $\ad(f)$ by construction. Similarly, using the isometry property, we have
	\begin{align*}
		\Gamma(B)a(f)g_1\os\cdots \os g_n
		&= \frac{1}{\sqrt{n}}\sum_{k=1}^n\braket{f,g_k}\Gamma(B)g_1\os\cdots\widehat{g_k}\cdots\os g_n\\
		&= \frac{1}{\sqrt{n}}\sum_{k=1}^n\braket{Bf,Bg_k}Bg_1\os\cdots\widehat{Bg_k}\cdots\os Bg_n\\
		& = a(Bf)Bg_1\os\cdots \os Bg_n = a (Bf)\Gamma(B)g_1\os\cdots\os g_n,
	\end{align*}
	which proves \subcref{appleampart:trans2}.
\end{proof}

\subsection[Q-Space Construction]{$\cQ$-Space Construction}\label{appsec:Q}

In this appendix, we define Gaussian processes indexed by a real Hilbert space $\fr$ on a probability space $(\cQ,\Sigma,\mu)$. We then recall the isomorphism theorem connecting $\FS(\fr\oplus\i\fr)$ and $L^2(\cQ)$. More details can be found in \cite{Simon.1974,LorincziHiroshimaBetz.2011}.

A {random process indexed by} $\fr$ is a ($\IR$-)linear map $\phi$ from $\fr$ to the random variables on $(\cQ,\Sigma,\mu)$. A {Gaussian random process indexed by} $\fr$ is a random process indexed by $\fr$, such that $\phi(v)$ is normally distributed with mean zero for any $v\in\fr$, has covariance
\begin{equation}\label{def:gauss}
	\int_{\cQ}\phi(v)\phi(w)d\mu = \frac 12\braket{v,w}_\fr
\end{equation}
and $\Sigma$ is the minimal $\sigma$-field generated by $\{\phi(v):v\in\fr\}$.

The following \lcnamecref{prop:gaussian} states existence and uniqueness of Hilbert space valued Gaussian processes. Extensive proofs can, for example, be found in \cite[Theorems I.6 and I.9]{Simon.1974} or \cite[Prop. 5.6, Section 5.4]{LorincziHiroshimaBetz.2011}. For the convenience of the reader, we add a sketch of the proof below.
\begin{lem}\label{prop:gaussian}
	For any real Hilbert space $\fr$ there exist a unique (up to isomorphism) probability space $(\cQ_\fr,\Sigma_\fr,\mu_\fr)$ and a unique (again up to isomorphism) Gaussian random process $\phi_\fr$ indexed by $\fr$ on $(\cQ_\fr,\Sigma_\fr,\mu_\fr)$.
\end{lem}
\begin{proof}[Sketch of Proof] \quad \\
\noindent 
{\it Existence}:
	We present one possible construction here, further constructions can be found in \cite{Simon.1974,LorincziHiroshimaBetz.2011}. Let $\{e_i\}_{i\in\sI}$ be a (not necessarily countable) orthonormal basis of $\fr$. We set $\cQ_\fr= \bigtimes_{i\in\sI}(\IR\cup\{\infty\})$ and equip it with the infinite product measure of the probability measures $\pi^{-1/2}\exp(-x_i^2)dx_i$, $i\in\sI$, which obviously is a probability measure itself. The Gaussian random process is now defined by $\phi_\fr(e_i)$ being the multiplication operator with the variable $x_i$. Clearly, $\mu_\fr\circ \phi_\fr(e_i)$ is normally distributed with mean zero and variance $\frac 12$. It also easily follows that $\int_{\cQ_\fr}\phi_\fr(e_i)\phi_\fr(e_j)d\mu_\fr = \frac 12\delta_{i,j}$, with $\delta$ denoting the usual Kronecker symbol.
	Finally, the Borel $\sigma$-algebra on $\cQ_\fr$ is generated by the set $\{\phi_\fr(e_i):i\in\sI\}$.
	Hence,
	 extending this definition to $\phi_\fr(f)$ for arbitrary $f\in\fr$ by linearity finishes the construction. 
	 
\noindent
{\it Uniqueness:} The uniqueness can be deduced from the Kolmogorov extension theorem \cite[Theorem 2.1]{Simon.1979}, which states that a probability space is uniquely determined by a consistent family of probability measures.
\end{proof} 
\noindent
The Hilbert space isomorphism introduced in \cref{prop:Q}, below, is often referred to as Wiener-It\^o-Segal isomorphism. More details on its construction, which is sketched below, can be found in \cite[Theorem I.11]{Simon.1974} or \cite[Prop. 5.7]{LorincziHiroshimaBetz.2011}. Here, we denote the complexification of the real Hilbert space $\fr$ as $\fr_\IC$, which is the real Hilbert space $\fr\times \fr$ with the complex structure given by
$ \i(\psi,\phi) = -(\phi,\psi)$.
\begin{lem}\label{prop:Q}
	There exists a unitary operator $\Theta_\fr:\FS(\fr_\IC)\to L^2(\cQ_\fr)$ such that
	\begin{enumlem}
		\item\label{prop:Q.1} $\Theta_\fr \Omega = 1$,
		\item\label{prop:Q.2} $\Theta_\fr^{-1} \phi_\fr (f)\Theta_\fr = \ph(v)$ for all $f\in\fr$.
	\end{enumlem}
\end{lem}
\begin{proof}[Sketch of Proof]
	\newcommand{\wick}[1]{\colon\! #1 \colon\!}
	We recursively define the Wick product of $\fr$-indexed Gaussian random variables by
	\[ \wick{\phi_\fr (f)}  = \phi_\fr(f), \quad \ \wick{\phi_\fr(f)\phi_\fr(f_1)\cdots \phi_\fr(f_{n})} = \phi_\fr(f)\wick{\phi_\fr(f_1)\cdots\phi_\fr(f_{n})} - \frac 12 \sum_{j=1}^{n}\braket{f,f_j}\wick{\prod_{i\ne j}\phi_\fr(f_i)}. \]
	Then, the map $\Theta_{\fr}:\FS(\fr_\IC)\mapsto L^2(\cQ_\fr)$ given by \subcref{prop:Q.1} and
	\[  f_1\os\cdots\os f_n \mapsto \sqrt 2\wick{\phi_\fr(f_1)\cdots \phi_\fr(f_n)} \qquad\mbox{for}\ f_1,\ldots,f_n\in\fr \]
	extends to a unitary. The property \subcref{prop:Q.2} follows explicitly from the definitions \cref{def:creann,def:field,def:gauss} and the fact that the pure symmetric tensors form a core for $\ph(f)$ by construction.
%	Further, we define normal ordering on Fock spaces, by taking all creation operators to the left and all annihilation operators to the right, without taking into account the canonical commutation relations. We denote normal ordering by the symbol $\fN$. The set $\operatorname{lin}\{\fN(\ph(f_1)\cdots \ph(f_n))\Omega:f_1,\ldots,f_n\in\fr,n\in\IN_0\}$ is dense in $\FS(\fr_\IC)$. Hence, we find that the operator $\Theta_{\fr}$ defined by \subcref{prop:Q.1} and
%	\[ \Theta_{\fr}\fN(\ph(f_1)\cdots \ph(f_n))\Omega = \wick{\phi_\fr(f_1)\cdots \phi_\fr(f_n)} \]
%	extends to a unitary operator from $\FS(\fr_\IC)\to L^2(\cQ_\fr)$.
%	Further, we define normal ordering on Fock spaces, by taking all creation operators to the left and all annihilation operators to the right, without taking into account the canonical commutation relations. We denote normal ordering by the symbol $\fN$. The set $\operatorname{lin}\{\fN(\ph(f_1)\cdots \ph(f_n))\Omega:f_1,\ldots,f_n\in\fr,n\in\IN_0\}$ is dense in $\FS(\fr_\IC)$. Hence, we find that the operator $\Theta_{\fr}$ defined by \subcref{prop:Q.1} and
%	\[ \Theta_{\fr}\fN(\ph(f_1)\cdots \ph(f_n))\Omega = \wick{\phi_\fr(f_1)\cdots \phi_\fr(f_n)} \]
%	extends to a unitary operator from $\FS(\fr_\IC)\to L^2(\cQ_\fr)$.
\end{proof}
%\benjamin{Ich habe hier versucht die Konstruktion so einfach und schnell wie möglich zu beschreiben. Sollten wir eventuell noch einen Remark schreiben, der erwähnt dass hier natürlich implizit die Normalordnung von Feldoperatoren auf dem Fockraum benutzt wird?}

\section{The Massive Spin Boson Model}\label{app:massive}
\newcommand{\HSBm}{H}
\newcommand{\Em}{E}

In this \lcnamecref{app:massive}, we prove that the ground state energy of the spin boson model with external magnetic field $H(\lambda,\mu)$ is in the discrete spectrum for any choice of the coupling constants $\lambda,\mu\in \IR$, if the dispersion relation $\omega$ is massive, i.e.,
\begin{equation}\label{eq:massive}
	m_\omega = \essinf_{k\in\IR^d}\omega(k) >0.
\end{equation}
The statement of the following \lcnamecref{thm:massive} for the case $\mu=0$, except for some simple technical restrictions on \cref{hyp:sbmin}, can be found in \cite{AraiHirokawa.1995}.
\begin{thm}\label{thm:massive}\label{thm:gsexistsHSBm}
	Assume \cref{hyp:sbmin,eq:massive} hold. Then $E(\lambda,\mu)\in\sigma_{\rm disc}(H(\lambda,\mu))$ for any choice of $\lambda,\mu\in\IR$.
\end{thm}
%\begin{rem}
%	The statement for the case $\mu=0$ can, for example, be obtained by combining the result in \cite{AraiHirokawa.1995} with \cref{prop:unique}.
%\end{rem}
We obtain the above \lcnamecref{thm:gsexistsHSBm} as a corollary of the following \lcnamecref{thm:HVZSB}.
\begin{prop}\label{thm:HVZSB} Assume \cref{hyp:sbmin} holds.
	Then, for all $\lambda,\mu\in\IR$, we have \[\inf \spess(\HSBm(\lambda,\mu))\ge\Em(\lambda,\mu)+m_\omega.\]
\end{prop}
\begin{rem}
	The statement can be seen as one half of an
	HVZ-type theorem for the spin boson model with external magnetic field.
	By a slight generalization of known techniques, one can prove
	\[ \spess(\HSBm(\lambda,\mu)) = [\Em(\lambda,\mu)+m_\omega,\infty), \]
	see for example \cite{AraiHirokawa.1995,DamMoller.2018a}.
	Here, we restrict our attention to the proof of the lower bound.
	%\benjamin{Weitere Zitate oder Bemerkung weglassen?}
\end{rem}

In the context of non-relativistic quantum field theory,
HVZ-type theorems are often proven using spatial localization of quantum particles, cf. \cite{DerezinskiGerard.1999,GriesemerLiebLoss.2001,Moller.2005,LossMiyaoSpohn.2007,HaslerSiebert.2020}.
%
%A prominently found proof of HVZ theorems in the literature uses localization estimates.
%Heuristically, the argument therein can be seen as follows:
%One first confines the bosons to a ball of radius $L$ in position space. As in the typical intuition of quantum mechanics, confined particles have discrete spectrum and to observe the essential spectrum one needs the presence of an unconfined particle.
%In the limit $L\to\infty$, the confined system behaves like the full Hamiltonian and hence the essential spectrum starts, when one free boson (which has at least the energy $m_\omega$) is added to the system.
%
%Applications of such localization techniques for the proof of HVZ theorems and hence for the existence of a ground state in the case of massive bosons can, for example, be found in \cite{DerezinskiGerard.1999,GriesemerLiebLoss.2001,Moller.2005,LossMiyaoSpohn.2007,HaslerSiebert.2020}.
%
%However, the use of localization estimates comes with a small downside. Explicitly, 
To bound the error terms obtained by confining the system to a ball of radius $L$ in position space, one needs to estimate the commutator of the multiplication operator $\omega$ and the Fourier multiplier $\eta(-\i\nabla/L)$, where $\eta$ is a smooth and compactly supported function. Bounds on the commutator can be easily obtained, when $\omega$ is Lipschitz-continuous (cf. \cite[Proof of Lemma 24]{HaslerSiebert.2020}). However, for less regular choices of the dispersion relation, a generalization of the standard localization approach does not seem obvious.
%\benjamin{Diese Bemerkung drin lassen? Ich finde sie für die Wahl des Beweises hilfreich.}

Here, we use an approach used by Fröhlich \cite{Frohlich.1974} and recently applied in \cite{DamMoller.2018a}, allowing us to work directly in momentum space and without any regularity assumptions on $\omega$ going beyond \cref{hyp:sbmin}.
%A similar approach is, for example, used in \cite{Frohlich.1974}.
 The proof needs several approximation steps, so we start out with a convergence \lcnamecref{lem:HSBmnormres}. In the statement, the norms $\|\cdot\|_\infty$ and $\|\cdot\|_2$ are the usual norms in $L^\infty(\IR^d)$ and $L^2(\IR^d)$.

\begin{lem}\label{lem:HSBmnormres}
	Let $(\omega_k)_{k\in \IN}$ and $(v_k)_{k\in\IN}$ be chosen such that \cref{hyp:sbmin} holds, i.e., 
 $\omega_k :\IR^d\to[0,\infty)$ is measurable and has positive values almost everywhere, and  $v_k \in L^2(\IR^d)$ satisfies $\omega_k^{-1/2}v_k\in L^2(\IR^d)$.
 Moreover,  define $\HSBm_k(\lambda,\mu)$ to be the operator defined in \cref{def:H}, i.e., 
\begin{equation} \label{defofHk} 
	H_k(\lambda,\mu) = \sigma_z\otimes\Id + \Id\otimes \dG(\omega_k) + \sigma_x\otimes(\lambda \ph(v_k)+\mu\Id) .
\end{equation}
 Further, assume $\| \omega / \omega_{k} \|_\infty$ and  $\| \omega_k / \omega \|_\infty$ are bounded and
 \begin{align}
 	&\lim_{k\to\infty}\left\|\frac{\omega_k}{\omega}-1\right\|_\infty = \lim_{k\to\infty}\left\|\frac{\omega}{\omega_k}-1\right\|_{\infty} = 0, \label{eq:uniformconv} \\
 	&\lim_{k\to\infty}\big\|v-v_k\big\|_2 = \lim_{k\to\infty}\big\|\omega^{-1/2}v-\omega_k^{-1/2}v_k\big\|_2 = 0. \label{eq:L2conv}
 \end{align}
	Then, for all $\lambda,\mu\in\IR$, the operators $\HSBm_k(\lambda,\mu)$ are uniformly bounded from below and converge to $\HSBm(\lambda,\mu)$ in the norm resolvent sense.
\end{lem}
\begin{rem}
	If $\omega$ and $\omega_k$ are uniformly bounded above and below by some positive constants, then the uniform convergence assumptions \cref{eq:uniformconv} are easily seen to be equivalent to $\|\omega_k-\omega\|_\infty \xrightarrow{k\to \infty}0$.
\end{rem}
\begin{proof}
	The uniform lower bound follows directly from the $L^2$-convergence assumptions \cref{eq:L2conv}, the bounds in  \cref{lem:Fockprops.standardbound} and the Kato-Rellich theorem \cite[Theorem X.12]{ReedSimon.1975}, see also the proof of \cref{lem:sbmin}.
	
	By the definition and setting $\omega_{\infty}=\omega$, for $n\in\IN$ and $f_1,\ldots,f_n\in L^2(\IR^d)$, we find
	\[ \left\|\dG(\omega_k) f_1\os\cdots\os f_n \right\| \le \left\|\frac{\omega_k}{\omega_{k'}}\right\|_\infty \left\|\dG(\omega_{k'})f_1\os\cdots\os f_n\right\| \qquad\mbox{for}\ k,k'\in\IN\cup\{\infty\}.\]
	Since the vectors $f_1\os\cdots\os f_n$ span a core for $\dG(\omega)$ by construction, we have $\cD(\dG(\omega_k))=\cD(\dG(\omega))$ for all $k\in\IN$. A similar argument yields
	\[ \left\|\left(\dG(\omega_k)-\dG(\omega)\right) \psi \right\| \le \left\|\frac{\omega_k}{\omega}-1\right\|_\infty\left\|\dG(\omega) \psi \right\| \qquad\mbox{for}\ \psi\in\cD(\dG(\omega)). \]
	Further, observe that the assumptions \cref{eq:uniformconv,eq:L2conv} easily imply
	\begin{equation} \label{eq:ccctg}
  \big
\|\omega^{-1/2}(v-v_k)\big\|_2 \leq  \big\|\omega^{-1/2}v- \omega_k^{-1/2} v_k\big\|_2 + \big\| (1 - \omega^{-1/2}/\omega_k^{-1/2}) \omega_k^{-1/2} v_k\big\|_2    \xrightarrow{k\to\infty} 0.  
\end{equation} 
	Now, by \cref{lem:sbmin}, $\big\| ( \Id \otimes \dG(\omega))(\HSBm(\lambda,\mu)+\i)^{-1}\big\|$ is bounded.
	Hence, using the resolvent identity as well as the standard bounds \cref{lem:Fockprops.standardbound} and \[\|(H_k(\lambda,\mu)+\i)^{-1}\|\le 1,\] we find  % using \eqref{eq:ccc}
	\begin{align*}
		\left\|\left(\HSBm_k(\lambda,\mu)+\i\right)^{-1} - \left(\HSBm(\lambda,\mu)+\i\right)^{-1}\right\| 
		\le &\, \left\|\frac{\omega_k}{\omega}-1\right\|_\infty \left\|(\Id\otimes \dG(\omega))\left(\HSBm(\lambda,\mu)+\i\right)^{-1}\right\|
		\\ & +  \sqrt 2|\lambda|\big\|\omega^{-1/2}(v_k-v)\big\|_2\| \sqrt{ \Id \otimes \dG(\omega)}\left(\HSBm(\lambda,\mu)+\i\right)^{-1}\| \\& + 2^{-1/2}|\lambda|\big\|v_k-v\big\|_2\left\|\left(\HSBm(\lambda,\mu)+\i\right)^{-1}\right\|\\
		&\xrightarrow{k\to\infty} 0, \qedhere
	\end{align*}
where the right hand side tends to zero by \cref{eq:uniformconv,,eq:L2conv,,eq:ccctg}.
\end{proof}

\begin{proof}[\textbf{Proof of \cref{thm:HVZSB}}]
	It suffices to treat the case $m_\omega>0$, since the statement is trivial otherwise. The proof has three steps and we fix $\lambda,\mu\in\IR$ throughout.
	
	\bigskip\noindent
	\hypertarget{proofHVZSB.1}{{\em Step 1.}} We first prove the statement in a very simplified case: Assume $M\subset \IR^d$ is a bounded and measurable set, $\omega\chr M$ and $v\chr M$ are simple functions on $M$ and $v=0$ almost everywhere on $M^c$.
	
	Let $M_k$ for $k=1,\ldots,N$ be a disjoint partition of $M$ into measurable sets such that $\omega\res{M_k}$ and $v\res{M_k}$ are constant for each $k=1,\ldots,N$. We define
	\[ \VS = \operatorname{lin} \{ \chr{M_k} : k=1,\ldots,N \}\subset L^2(\IR^d).  \]
	Since $\VS$ is finite-dimensional, it is closed and we have the decomposition $L^2(\IR^d)=\VS\oplus \VS^\perp$. Observing that by the assumptions $v\in \VS$, we can define
	\[  T = \sigma_z \otimes \Id + \Id \otimes\dG(\omega) + \sigma_x\otimes(\lambda\ph(v) + \mu \Id) \quad\mbox{as operator on}\ \IC^2\otimes\FS(\VS).  \]
	We define a linear map  $U:\IC^2\otimes \FS \to \bigoplus\limits_{n=0}^\infty (\IC^2\otimes \FS(\VS))\otimes \left(\VS^\perp\right)^{\otimes_\sfs n}$, where we set $\IC := V^{\otimes_\sfs 0}$ for any vector space $V$,  by
%	\[ \alpha \otimes (f_1\os \cdots \os f_n) \mapsto \bigoplus_{\substack{(I_1,I_2)=\fP\in \cP_n\\|\fP|=2}} \alpha \otimes (P_\VS f_1\os \cdots \os P_\vs f_n)\otimes (1-P_\VS)f_1\os\cdots\os(1-P_\VS)f_n), \]
%*******************************
	\[  (f_1\os \cdots \os f_n) \mapsto \bigoplus_{ \underline{p} \in \{ 0 , 1 \}^n }   \bigg( \bigotimes_{\substack{ j =1,\ldots,n \\ p_j =1 } }^\sfs P_\VS f_j   \bigg) 
\otimes \bigg(  \bigotimes^\sfs_{\substack{ j =1,\ldots,n \\ p_j =0 }} (1- P_\VS) f_j \bigg)  \]
	where $P_\VS$ denotes the orthogonal projection in  $L^2(\IR^d)$ onto $\VS$.
It is straightforward to verify that $U$ is unitary.
 For $n\in\IN_0$, we denote by $\Pi_n$ the projection onto the subspace $\IC^2\otimes \FS(\VS)\otimes\left(\VS^\perp\right)^{\otimes_\sfs n}$ in the range of $U$.
	
	From the definition \cref{def:H} and the definitions of $T$ and $U$ above, it is easily verified that
	\begin{equation}\label{eq:directsumdecomp}
		U\HSBm(\lambda,\mu)U^* = T \oplus \bigoplus_{n=1}^\infty \left(T\otimes \Id_{\left(\VS^\perp\right)^{\otimes_\sfs n}} + \Id \otimes \dG(\omega)\res{\left(\VS^\perp\right)^{\otimes_\sfs n}} \right).
	\end{equation}
	Thus $\inf \sigma(T)\ge \Em(\lambda,\mu)$ and hence $\Pi_k UH(\lambda,\mu)U^* \Pi_k \ge E(\lambda,\mu) + m_\omega$ for $k\ge 1$.
	
	Now, assume $\gamma\in \spess(\HSBm(\lambda,\mu))$. Then, by Weyl's criterion, there exists a normalized sequence $(\psi_n)_{n\in\IN}$ weakly converging to zero such that
	\begin{equation}\label{eq:Weylsequence}
		\lim\limits_{n\to\infty} \left\| \left(\HSBm(\lambda,\mu)-\gamma\right)\psi_n \right\| = 0.
	\end{equation}
	Using \cref{eq:Weylsequence} and    \cref{eq:directsumdecomp}, we find 
	\begin{equation}\label{eq:Weylineq0}
  \gamma  = \lim_{n \to \infty} 		\langle \psi_n    , H(\lambda, \mu) \psi_n   \rangle  \ge \Em(\lambda,\mu) + m_\omega +  \limsup_{n\to\infty}\Braket{\Pi_0(U\psi_n),T\Pi_0(U\psi_n)} .
	\end{equation}
	We now want to show that the last term converges to zero. To that end,  we write 
	\begin{equation}\label{eq:Weylineq}
 \Braket{\Pi_0(U\psi_n),T\Pi_0(U\psi_n)}  =  \Braket{S\Pi_0(U\psi_n),S^{-1}T\Pi_0(U\psi_n)},
	\end{equation}
	where $S=\Id\otimes (\Id + \dG(\omega))$ as operator on $\IC^2\otimes \FS(\VS)$.
	By \cref{eq:Weylsequence}, $\|S \Pi_0(U^*\psi_n)\|$ is uniformly bounded in $n$.
	We write
	\[\FS^{(\le N)}(\VS) = \bigoplus_{n=0}^N \os^n\VS.\] 
	The assumption $\omega\ge m_\omega>0$ implies that 
	\[\lim_{N\to\infty} S^{-1}T \res{\FS^{(\le N)}(\VS)} = S^{-1}T.\]
	Since $\FS^{(\le N)}(\VS)$ is finite-dimensional by construction $S^{-1}T \res{\FS^{(\le N)}(\VS)} $ has finite rank for any $N\in\IN$ and it follows that $S^{-1}T$ is compact, since it is the limit of compact operators. Hence, the last term on the right hand side of \cref{eq:Weylineq} (and hence that of \cref{eq:Weylineq0}) converges to zero as $n\to\infty$.
	
	This finishes the first step.\\[.5em] %$\lozenge$
	\noindent
	\hypertarget{proofHVZSB.2}{{\em Step 2.}} We now relax the condition that $\omega$ and $v$ must be simple functions: Assume $M\subset \IR$ is a bounded measurable set, $\omega\chr M$ is bounded and $v=0$ almost everywhere on $M^c$.
	
	By the simple function approximation theorem \cite[Theorem 2.10]{Folland.1999}, we can pick a sequence $(\omega_k)_{k\in\IN}$ of pointwise monotonically increasing simple functions on $M$ uniformly converging to $\omega$. Outside of $M$ we set $\omega_k$ equal to $\omega$. Further, w.l.o.g., we can assume that $m_\omega \le \omega_k$.
	% there exist constants $a,b>0$ such that $\omega,\omega_k\in[a,b]$ holds on $M$, by the assumptions that $m_\omega>0$ and $\omega$ is bounded on $M$.
	
	For given $k\in \IN$ let $M_{k,i}$, $i=1,\ldots,N_k$, be a disjoint partition of $M$ into measurable sets such that $\omega_k\res{M_{k,i}}$ is constant for all $i=1,\ldots,N_k$. Further, w.l.o.g, we can assume that
	\begin{equation}\label{eq:diam}
		\min_{i=1,\ldots,N_k}\operatorname{diam}(M_{k,i}) \xrightarrow{k\to\infty} 0,
	\end{equation}
	where $\operatorname{diam}$ denotes the usual diameter of a bounded set.
	Then, we define a projection $P$ onto the vector space of  simple functions with support in $M$ by
	\[  P_k f = \sum_{i=1}^{N_k} \frac{\chr{M_{k,i}}}{\operatorname{vol}(M_{k,i})}\int_{M_{k,i}} f(x)\d x, \]
	which can be easily verified to be well-defined for any $f\in L^2(\IR^d)$.
	If $f$ is continuous and compactly supported on $M$, then it is straightforward to verify $P_kf\xrightarrow{k\to\infty} f$ in $L^2$-sense. Since the continuous, compactly supported functions are dense in $L^2(M)$, this implies $\slim\limits_{k\to\infty} P_k = \Id_{L^2(M)}$ in strong operator toplogy. % $L^2(M)$.
	
	We now define
	\[ v_k = \omega_k^{1/2}P_k(\omega^{-1/2}v) \]
	and observe this directly implies $\omega_k^{-1/2} v_k$ converges to $\omega^{-1/2}v$ in $L^2$-sense. Further,
	by the triangle inequality and monotonicity of the integral we find 
	\begin{align*}
		\|v_k-v\|_2^2 & \le \int_M |v_k-\omega_k^{1/2}\omega^{-1/2}v|^2 + \int_M |\omega_k^{1/2}\omega^{-1/2}v-v|^2 \\
		& \le \|\omega\res M \!\!\|_\infty \|\omega_k^{-1/2}v_k-\omega^{-1/2}v\|_2^2 + \frac 1{m_\omega}\|\omega_k^{1/2}-\omega^{1/2}\|_\infty \|v\|_2^2.
	\end{align*}
	By construction the right hand side goes to zero as $k\to\infty$.
	
	Hence,  we have shown that all assumptions of \cref{lem:HSBmnormres} are satisfied and the operators $\HSBm_k(\lambda,\mu)$, defined in \eqref{defofHk},
   are uniformly bounded from  below and converge to $\HSBm(\lambda,\mu)$ in  norm resolvent sense.
	Further,  $\omega_k$ and $v_k$ satisfy by construction  the assumptions of \hyperlink{proofHVZSB.1}{Step 1}.
	The statement  of the Theorem now follows under the simplifying assumptions of Step 2, since on the one hand  the uniform convergence of $\omega_k$ to $\omega$ implies $m_{\omega_k}$ converges to $m_\omega$ and on the other  the norm resolvent convergence and uniform lower boundedness imply convergence of the ground state energy,   
%\cite{ReedSimon.1972}[??] 
 cf.  \cite[Theorem 6.38]{Teschl.2014}, as well as  the infimum of the essential spectrum, cf.
\cite[Theorem XIII.77]{ReedSimon.1978}.\\[.5em]  %$\lozenge$
	\noindent
	{\em Step 3.} We now move to the general case.
	
	For  $k \in \N$  define
	\[ M_k = \{q\in \R^d  : \ | q| \leq k ,  \ v(q)\ne 0 , \ \omega(q) \leq k \} \cup \{q\in\IR^d:v(q)=0\}. \]
%	\[ M_R = \left(\{k\in\IR^d:v(k)\ne 0\}\cap \{k\in\IR^d:\omega(k)<R\}\cap B_R(0)\right)\cup \{k\in\IR^d:v(k)=0\}. \]
	Set $v_k = \chr{M_k}v$. Then, taking $k\to\infty$, it is straightforward to verify that both $v_k$ and $\omega^{-1/2}v_k$ converge to $v$ and $\omega^{-1/2}v$ in $L^2$-sense, respectively. Hence, we can once more apply \cref{lem:HSBmnormres} to see that $\HSBm_k(\lambda,\mu)$ with  $\omega = \omega_k$ and $v_k$ in \cref{defofHk}  is uniformly bounded below and converges to $\HSBm(\lambda,\mu)$ in the norm resolvent sense as $k\to\infty$. Since, $\omega$ and $v_k$ also satisfy the assumptions of \hyperlink{proofHVZSB.2}{Step 2}, the statement again follows by the spectral convergence as in Step 2. % $\lozenge$
\end{proof}
We conclude this \lcnamecref{app:massive} with the\vspace*{-.5em}
\begin{proof}[\textbf{Proof of \cref{thm:gsexistsHSBm}}]
	Since the spectrum of $H(\lambda,\mu)$   is the disjoint union of discrete and essential spectrum, the statement follows from \cref{thm:HVZSB}.
\end{proof}

\vspace*{-2em}\section*{Declarations}\vspace*{-.5em}

The authors have no relevant financial or non-financial interests to disclose.\\[-2.5em]

\bibliographystyle{halpha-abbrv}
\bibliography{lit}

\end{document}